\documentclass[11pt]{article}
\usepackage[utf8]{inputenc}
\usepackage[T1]{fontenc}
\usepackage[dvipsnames]{xcolor}
\usepackage{url}
\usepackage{float}
\usepackage{booktabs}
\usepackage{amsfonts}
\usepackage{amsmath}
\usepackage{amssymb}
\usepackage{nicefrac}
\usepackage{tikz}
\usepackage{pgfplots}
\usepgfplotslibrary{patchplots}
\pgfplotsset{compat=1.18}
\usepackage{microtype}
\usepackage{tcolorbox}
\usepackage{enumitem}
\usepackage{graphicx}
\usepackage[ruled]{algorithm2e}

\usepackage[switch]{lineno}
\usepackage[section]{placeins}
\usepackage{dirtytalk}
\usepackage{mathtools}
\usepackage{amsthm}
\usepackage{authblk}
\usepackage{natbib}
\usepackage{geometry}
\usepackage[colorlinks = true, linkcolor=NavyBlue,citecolor=ForestGreen]{hyperref}
\usepackage[nameinlink]{cleveref}

\newtheorem{theorem}{Theorem}[section]

\newtheorem{lemma}[theorem]{Lemma}
\newtheorem{corollary}[theorem]{Corollary}
\theoremstyle{definition}
\newtheorem{definition}[theorem]{Definition}
\newtheorem{assumption}[theorem]{Assumption}
\theoremstyle{remark}
\newtheorem{remark}[theorem]{Remark}
\newtheorem{claim}{Claim}
\newtheorem{fact}{Fact}

\crefname{remark}{remark}{remarks}
\Crefname{remark}{Remark}{Remarks}

\title{Improved Last-iterate Convergence Properties\\ for the FLBR-MWU Dynamics}

\author[1]{Michail Fasoulakis}
\author[2,3,4]{Evangelos Markakis}
\author[5]{Giorgos Roussakis}
\author[2,3]{Christodoulos Santorinaios}

\affil[1]{Royal Holloway, University of London, UK}
\affil[2]{Athens University of Economics and Business, Greece}
\affil[3]{Archimedes/Athena RC, Greece}
\affil[4]{Input Output Group (IOG), Greece}
\affil[5]{Foundation for Research and Technology--Hellas, Greece}

\hypersetup{
pdftitle={Improved Last-iterate Convergence Properties for the FLBR-MWU Dynamics},
pdfsubject={cs.GT},
pdfauthor={Micheal Fasoulakis, Evangelos Markakis, Giorgos Roussakis, Christodoulos Santorinaios},
pdfkeywords={FLBR-MWU, last-iterate},
}

\begin{document}
\maketitle

\begin{abstract}
	We revisit a variant of Multiplicative Weights Update (MWU), defined recently by \citet{FMPV22}, and denoted as Forward Looking Best Response MWU (FLBR-MWU). These dynamics are based on the approach of extra-gradient methods, with the tweak of using  different learning rates in the intermediate step and the actual update step. So far, it has been proved that this algorithm attains asymptotic last-iterate convergence but no explicit rate has been known. We answer the open question from Fasoulakis et al. by establishing a concrete convergence rate for the duality gap. In particular, we show a geometric convergence rate, of the form $O(c^t)$, where $c<1$ is independent of time but dependent on game parameters, such as the maximum eigenvalue of the Jacobian matrix.
We also complement our theoretical analysis with an experimental comparison to OGDA (Optimistic Gradient Descent-Ascent), which ranks among the best last-iterate methods for solving zero-sum games. We demonstrate that the performance of the FLBR-MWU method matches or, in some cases, outperforms OGDA.
\end{abstract}

\newpage

\section{Introduction}
Our work focuses on learning algorithms with convergence guarantees in two-player bilinear zero-sum games. 
This is by now an extensively studied domain, spanning a few decades of research progress already. Given a game described by its payoff matrix, what we are after here is algorithms that eventually reach a Nash equilibrium, from which no player has an incentive to deviate. Some of the earlier and standard results in this area concern convergence {\it on average}. I.e., it has long been known that by using no-regret algorithms, the empirical average of the players' strategies over time converges to a Nash equilibrium in zero-sum games and to more relaxed equilibrium notions (coarse correlated equilibria) for general games \citep{freund1999adaptive}.

In recent years, the attention of the relevant community has gradually shifted from convergence on average to the more robust notion of {\it  last-iterate convergence}, a property highly desirable from an application perspective. This means that the strategy profile $(x^t,y^t)$, reached at iteration $t$ of an iterative algorithm, converges to the actual equilibrium as $ t \to \infty$. 
Unfortunately, many of the initially developed methods do not satisfy this property. No-regret algorithms, like the Multiplicative Weights Update (MWU) method, are known to converge only in an average sense.
In fact, it was shown in \citet{BP18,MPP18} that several MWU variants do not satisfy last-iterate convergence.

Motivated by these considerations, the last decade has seen a series of works studying last-iterate convergence. The majority of these works have focused on the fundamental class of zero-sum games. Zero-sum games have played an important role in the development of game theory and optimization, and more recently, there has also been a renewed interest, given their relevance in formulating GANs in deep learning \citep{GPMXWOCB14}. The positive results that have been obtained for zero-sum games show that improved variants of Gradient Descent such as 
the Optimistic Gradient Descent/Ascent method (OGDA), or the Extra-Gradient method (EG) attain last iterate convergence. Several other methods have also been obtained and compared to each other with respect to their convergence rate. Overall, one can say that we now have a much better understanding of the learning dynamics that converge in zero-sum games.

Despite the positive progress, however, several important questions still remain unanswered. First, it is often difficult to have tight bounds in analyzing such learning algorithms. Furthermore, even for bilinear, zero-sum games, the best attainable rate of convergence is not yet fully understood. The currently best rate that is applicable to all such games is $O({1}/{\sqrt{t}})$ in terms of the duality gap \citep{COZ22,GorbunovTG22}, where the hidden terms in the $O(\cdot)$ notation depend on the game dimension but not on the payoff matrix. In fact this also holds for the more general class of convex-concave min-max optimization problems. 
It is conceivable though that better rates could be achieved for bilinear games. The work of \citet{Wei2021LinearLC} establishes a geometric convergence rate of $O(c^t)$ ($c<1$) for the OMWU method, discussed further in the sequel, albeit with game-dependent parameters within the  $O(\cdot)$ term. It remains an open problem whether a geometric convergence rate can be achieved where the dependence is only on the game dimension.

\subsection{Our contributions}

We revisit a promising variant of MWU, defined recently in \citet{FMPV22} and denoted as Forward Looking Best-Response Multiplicative Updates (FLBR-MWU), and referred to from now on as FLBR for brevity. The dynamics are based on the approach of extra-gradient methods, with the tweak of using a different and more aggressive learning rate in the intermediate step. 
Our main contributions can be summarized as follows: 
\begin{itemize}[noitemsep,leftmargin=*]
    \item So far, it was only known that the FLBR algorithm attains asymptotic last-iterate convergence for bilinear zero-sum games, but without any explicit rate. We answer the open question from \citet{FMPV22} by establishing concrete rates of convergence. Using the duality gap as our metric, we first show a geometric rate, of the form $O(c^t)$, until we reach an approximate Nash equilibrium, for an appropriate level of approximation. More precisely, the parameter $c< 1$ is independent of the time $t$, but depends on the game and on its dimension. 
    \item For games with a unique Nash equilibrium, we further prove that once we reach an approximate equilibrium, the duality gap keeps getting decreased with a geometric rate, until the exact equilibrium solution. The caveat here is that there is again a dependence on the game parameters, namely on the Jacobian matrix evaluated at the equilibrium. An analogous result also holds for the OMWU method \citep{Wei2021LinearLC}, which is a different variant of MWU, as mentioned earlier, but w.r.t. the KL divergence, and with a different dependence on the game parameters. We view as advantages of our analysis that it yields a simpler, shorter and more intuitive proof compared to the lengthier analysis for OMWU in \citet{Wei2021LinearLC}, 
    
    Interestingly, our proof utilizes some ideas from the analysis of the Arimoto-Blahut algorithm (for computing Shannon’s capacity of a discrete memoryless channel), highlighting connections to a neighboring field.
    \item We also investigate further properties of FLBR. In  \Cref{sec:regret}, we prove that it is not a no-regret algorithm, which was not known before. At the same time, in our concluding section, we explore aspects of {\it forgetfulness}, as introduced recently in \citet{caiforget}. We show that in contrast to OMWU, FLBR seems to exhibit forgetfulness, which serves as an indication for fast performance.
    \item Finally, we perform an experimental comparison of FLBR against OGDA, which is among the best known methods for solving zero-sum games\footnote{We focus on the comparison against OGDA and not OMWU since the latter is not as competitive in practice (observed also in other recent works).}. The results reveal that FLBR is generally competitive with OGDA, and in some cases outperforms it.
    \end{itemize}
Overall, we believe that our work provides a more complete treatment on the power and limitations of the FLBR method for bilinear games. The main message that emerges from our work is that even though the theoretical analysis of FLBR is along the same spirit as OMWU (i.e., having a game-dependent geometric rate), its experimental performance is superior and creates potential for future practical adaptation  (in contrast to OMWU, which is known not to perform well in practice).

\subsection{Related work}
There is a vast literature on solving zero-sum games. We focus here on the most relevant works to ours and provide further discussion in \Cref{appsec:relwork}.

Within the last years, there has been great interest in designing fast learning algorithms for zero-sum games. Although this direction started several decades ago, e.g. with the fictitious play algorithm \citep{brown1951iterative,R51}, it has received significant attention more recently given the relevance to formulating GANs in deep learning \citep{GPMXWOCB14} and also other applications in machine learning. 
Some of the earlier and standard results in this area concern convergence {\it on average}. That is, it has been known that by using no-regret algorithms, such as the Multiplicative Weights Update (MWU) methods \citep{AHK12}, the empirical average of the players' strategies over time converges to a Nash equilibrium in zero-sum games. Similarly, one could also utilize Gradient Descent/Ascent (GDA) algorithms.
Several other algorithms for zero-sum games are built within the framework of regret minimization both in theory \citep{CJST19,CJJS24} and in applications \citep{FarinaKS21}.

Coming closer to our work, within the last decade, there has also been a great interest in algorithms attaining the more robust notion of {\it  last-iterate convergence}. This means that the strategy profile $(x^t,y^t)$, reached at iteration $t$, converges to the actual equilibrium as $ t \to \infty$. 
Negative results in \citet{BP18} and \citet{MPP18} show that several no-regret algorithms, such as many MWU as well as GDA variants, do not satisfy last-iterate convergence. Instead, they may diverge or enter a limit cycle. Motivated by this, there has been a series of works on obtaining algorithms with provable last iterate convergence. The positive results that have been obtained for zero-sum games show that improved versions of Gradient Descent such as the Extra-Gradient method \citep{K76} or the Optimistic Gradient method \citep{P80} attain last-iterate convergence. In particular, \citet{daskalakis2018training} and \citet{DBLP:conf/aistats/LiangS19} show that the optimistic variant of GDA (referred to as OGDA) converges for zero-sum games. Analogously, OMWU (the optimistic version of MWU) also attains last-iterate convergence, as shown in \citet{Daskalakis2019LastIterateCZ} and further analyzed in \citet{Wei2021LinearLC}. Further approaches with convergence guarantees have also been proposed, such as primal-dual hybrid gradient methods \citep{LY23}.
For the case of constrained bilinear zero-sum games, the best convergence rate for the duality gap achieved so far is by \citet{COZ22,GorbunovTG22}, which is $O({1}/{\sqrt{t}})$. We note that better rates are achievable for the case of unconstrained bilinear zero-sum games, as e.g., in \citet{mokhtari2020unified}, but this is an easier problem than what we focus on here.
We also note that for the metric of KL divergence, \citet{Wei2021LinearLC} provide a geometric rate, which is dependent on game parameters.

The method we analyze here is inspired by the general approach of extra-gradient methods, but with the tweak of using different learning rates in the intermediate and final step of each iteration. The idea of using different rates in these two steps of each iteration has also been successful in other recent works. It has been used in \citet{Azizian} for a model that concerns the unconstrained bilinear case.
Again for the unconstrained case (but even beyond convex-concave functions), the work of \citet{DBLP:conf/aistats/DiakonikolasDJ21} showed how the use of different learning rates achieves convergence guarantees for their method (referred to as EG+). These ideas have also been applied successfully in the stochastic setting, under noisy gradient feedback,  \citep{HsiehIMM20}.

Several of these methods have also been studied beyond bilinear payoff functions or beyond zero-sum games, including \citep{GPD20} and also \citep{DBLP:conf/aistats/DiakonikolasDJ21} where positive results are shown for a class of non-convex and non-concave problems. 
There are also negative results however as e.g., established in \citet{DSZ21}. Going beyond min-max problems, the work of \citet{PP24} obtains last-iterate convergence rates in rank-1 games. Results for richer classes of games are provided in \citet{anagnostides2022last}, including potential and constant-sum polymatrix games. The landscape, however, is overall less clear. 

\section{Preliminaries}\label{sec:prelim}
We consider two-player $n \times n$ zero-sum games $(R,-R)$. Without loss of generality, we consider that $R\in [0,1]^{n \times n}$ is the payoff matrix of the row player, and $-R$ is the payoff matrix of the column player.\footnote{Any game can be transformed to a game with entries in the interval $[0,1]$ with the same Nash equilibria.}
A mixed strategy is a probability distribution $x = (x_1,\dots, x_n)^\top $ over the standard simplex $\Delta_n$, 
We also denote by $e_i$ the distribution corresponding to a pure strategy $i$, with 1 in the index $i$ and zero elsewhere. 
The support of a mixed strategy $x$ is the set of the pure strategies to which $x$ assigns positive mass, i.e. $supp(x) = \{i| x_i > 0\}$.

A strategy profile is a tuple $(x, y)$, where $x$ (resp. $y$) is the strategy of the row (resp. column) player. Given a profile $(x,y)$, the expected payoff of the row (resp. column) player is $x^\top  Ry$ (resp $-x^\top  Ry$).  

\begin{definition}[$\varepsilon$-Nash equilibrium ($\varepsilon$-NE)]\label[definition]{def:approx-Nash}
A strategy profile $(x,y)$ is an $\varepsilon$-Nash equilibrium of the game $(R,-R)$, with $R\in [0,1]^{n \times n}$, for $\varepsilon \in [0,1]$, if and only if, for any $i,j\in [n]$,
\[
{x}^\top  Ry +\varepsilon\geq e_i^\top Ry, \text{ and } {x}^\top  Ry -\varepsilon\leq {x}^\top  Re_j.
\]
\end{definition}
\noindent By setting $\varepsilon=0$ we have an exact NE.
Another way to view $\epsilon$-NE is that the players are playing approximate best response strategies. 
A strategy $x$ is a $\rho$-best-response strategy against $y$ for the row player, for $\rho \in [0,1]$, if and only if $x^\top R y +\rho \geq e^\top_j R y$, for any $j\in [n]$. Similarly, a strategy $y$ for the column player is a $\rho$-best-response strategy against some strategy $x$ of the row player if and only if $x^T R y \leq x^T R e_i + \rho$, for any $i\in [n]$. 

Next we will define our progress measure.
\begin{definition}[Duality Gap]\label[definition]{def:V}
For zero-sum games, the duality gap function $V$ is defined as 
\[ V(x,y) = \max_i e_i^\top Ry-\min_j x^\top  Re_j. \]
\end{definition}
The duality gap is a central notion in game theory as it captures the combined loss of the players for not employing best responses and hence for deviating from a NE, as seen in the fact below.
\begin{fact}\label[fact]{Nash from V}
A strategy profile $(x^*,y^*)$ is a Nash equilibrium of a zero-sum game if and only if it is a (global) minimum of the function $V(x,y)$. Furthermore, if $V(x,y) \leq \varepsilon$, then $(x,y)$ is an $\varepsilon$-NE.
\end{fact}
Before proceeding with the dynamics, we state a simple lemma that relates the $L_1$ norm with the duality gap function.
\begin{lemma}\label[lemma]{lem:norm to Duality gap}
For any $x, y$ it holds that $\max_i e^\top _iRy  \leq ||y-y^*||_1 + v$ and $\min_j x^\top  Re_j  \geq ||x-x^*||_1 + v$, where $v$ is the value of the zero-sum game w.r.t. the row player.

\end{lemma}

\subsection{FLBR-MWU dynamics}

Here we restate the Forward Looking Best-Response Dynamics as introduced in \citet{FMPV22}. These dynamics follow an extra-gradient approach to find a Nash Equilibrium. Specifically, each iteration involves an intermediate step that serves as a prediction for the update step. The difference with other extra-gradient-like approaches is that different learning rates are used in the intermediate and the final step, which appears crucial to the effectiveness of this approach.

Given an initial strategy profile $(x^0,y^0)$, the two steps of the dynamics can be described as follows:

    \begin{align*}
\text{Intermediate step:\ } &{\hat{x}_i}^{t} = x_i^{t-1}\cdot \frac{e^{\xi \cdot e_i^\top  R{y}^{t-1} }}{\sum\limits_{j} x_j^{t-1}\cdot e^{\xi \cdot e_j^\top  R{y}^{t-1}}},\\  &  {\hat{y}}_j^{t} = y_j^{t-1}\cdot \frac{e^{-\xi \cdot  e_j^\top  R^\top  {x}^{t-1}}}{\sum\limits_{i} {y}_i^{t-1}\cdot e^{-\xi \cdot e_i^\top  R^\top  {x}^{t-1}}} \\
\text{Update step:\ } &x_i^t = x_i^{t-1}\cdot \frac{e^{\eta \cdot e_i^\top  R\hat{y}^{t}}}{\sum\limits_{j} x_j^{t-1}\cdot e^{\eta \cdot  e_j^\top  R\hat{y}^{t} }}, \\ &y_j^t = y_j^{t-1}\cdot \frac{e^{-\eta \cdot e_j^\top  R^\top  \hat{x}^{t} }}{\sum\limits_{i} {y}_i^{t-1}\cdot e^{-\eta \cdot e_i^\top  R^\top  \hat{x}^{t} }}.
\end{align*}

When $\xi=\eta$ in the above steps, this is referred to as Mirror-Prox in \cite{nemirovski2004prox}. Contrary to the conventional wisdom of using rather small learning rates to ensure contraction, our approach utilizes a large value for $\xi$ (aggressive rate for the intermediate exploration step) coupled with a small (conservative) learning rate $\eta \in (0,1)$ for the update step.   
Finally, we state an important property that we will use at various points in the sequel: 

\begin{lemma}[\citet{FMPV22}]\label[lemma]{fact:best-response}
    For any $t>0$, it holds that as $\xi \to \infty$, $\hat{x}^t$ (resp. $\hat{y}^t$) converges to the proportional best response strategy $\Tilde{x}^t$ against $y^{t-1}$ (resp. $\Tilde{y}^t$ against $x^{t-1}$). Specifically, if $B_x$ (resp. $B_y$) is the set of pure best responses against $y^{t-1}$ (resp. $x^{t-1}$), we have
    \[
    \Tilde{x}_i^t = \begin{cases}
        \frac{x_i^{t-1}}{\sum\limits_{j \in B_x} x_j^{t-1}},\ &i \in B_x, \\ \quad 0, &i \not\in B_x,      \end{cases} 
        \text{ and } \Tilde{y}_i^t = \begin{cases}
        \frac{y_i^{t-1}}{\sum\limits_{j \in B_y} y_j^{t-1}},\ &i \in B_y, \\ \quad 0, &i \not\in B_y. 
    \end{cases} 
    \]
\end{lemma}

\begin{assumption}\label[assumption]{assumption:1}
We will start the dynamics from the fully uniform distribution, i.e., $x^0 = y^0 = (1/n,\dots, 1/n)$. Furthermore, we will use a fixed $\eta$, independent of $t$ in all iterations.\end{assumption}

\section{Convergence analysis}\label{sec:conv}

In this section, we use the duality gap as a metric to study the rate of convergence for FLBR-MWU. This answers the question left open by \citet{FMPV22}. Our analysis consists of two main parts. First, we obtain an analysis of convergence until an appropriate approximate equilibrium is reached, where the degree of approximation depends on $\eta$. Then, we show that if $\eta$ is sufficiently small, so as to guarantee that we are close to the exact solution, we can maintain a geometric rate to the exact equilibrium. The analysis comes at the cost of introducing a dependency on the game parameters. Finally, we note that all missing proofs in the sequel can be found in the Appendix.

\subsection{Convergence to an approximate equilibrium}
\label{subsec:approx}
Let $(x^*,y^*)$ be an arbitrary exact Nash equilibrium, and let $(x^t,y^t)$ be the strategy profile produced by the dynamics at the end of time step $t$. We stress that for the convergence to an approximate equilibrium, we do not need to assume uniqueness. 

In our analysis, we use the \textit{Kullback-Leibler (KL)} divergence of a profile from $(x^*,y^*)$, defined as 
\begin{gather*}
 D_{KL}((x^*,y^*)||(x^t,y^t)) =
\sum\nolimits_{i=1}^{n}x^*_i \cdot  \ln(x^*_i/x^t_i) + \sum\nolimits_{j=1}^{n}y^*_j \cdot \ln(y^*_j/y^t_j).
\end{gather*}
Note that by the definition of the dynamics, $x_i^t$ and $y_j^t$ are always positive for any $i, j$ and $t$; hence the ratios above are well-defined. For brevity, we write $D_{KL}((x^*,y^*)||(x^t,y^t))$ as $D^t$. The main technical property for the analysis of reaching an approximate equilibrium is the following lemma.  

\begin{lemma}\label[lemma]{lem:V-KL}
It holds that for any $t\geq 1$, and any $\eta\leq 1/2$
\[
\eta \cdot \left[ (\hat{x}^{t})^\top  R{y}^{t-1} - (x^{t-1})^\top  R\hat{y}^t \right] \leq  D^{t-1} - D^t  + 4\eta^2.
\]
\end{lemma} 
\begin{proof}
We first rewrite the KL terms, by using the definition of the dynamics. 
\begin{align*}
&D_{KL}((x^*,y^*)||(x^{t-1},y^{t-1}))   -D_{KL}((x^*,y^*)||(x^{t},y^{t})) \\
&= \sum\nolimits_{i=1}^{n} x^*_i \cdot \ln(x^{t}_i/x^{t-1}_i) + \sum\nolimits_{j=1}^{n} y^*_j \cdot \ln(y^{t}_j/y^{t-1}_j) \\
&= \sum\limits_{i=1}^{n}x^*_i \cdot \ln e^{\eta \cdot e_i^\top  R \hat{y}^t} - \ln\Big(\sum\limits_{k=1}^{n} x^{t-1}_k \cdot e^{\eta \cdot e_k^\top  R \hat{y}^t}\Big) \\
&\hphantom{=} + \sum\limits_{j=1}^{n}y^*_j \cdot \ln e^{-\eta \cdot e_j^\top  R^\top  \hat{x}^t} -\ln\Big(\sum\limits_{k=1}^{n} y^{t-1}_k \cdot e^{-\eta \cdot e_k^\top  R^\top  \hat{x}^t}\Big) \\
& = \eta \cdot (x^*)^{T}R\hat{y}^t - \eta \cdot (y^{*})^T R^\top  \hat{x}^t - \ln\Big(\sum\limits_{k=1}^{n} x^{t-1}_k\cdot e^{\eta \cdot e_k^\top  R \hat{y}^t}\Big) - \ln\Big(\sum\limits_{k=1}^{n} y^{t-1}_k \cdot e^{-\eta \cdot  e_k^\top  R^\top  \hat{x}^t}\Big).
\end{align*}

We now use the Taylor expansion of the exponential function in the arguments of the last two logarithms. 
For the first logarithmic term, this becomes:
\begin{align*}
\ln\Big(\sum\limits_{k=1}^{n} x^{t-1}_k\cdot e^{\eta \cdot e_k^\top  R \hat{y}^t}\Big) &= \ln\Big(1 + \eta \cdot (x^{t-1})^\top  R\hat{y}^t + \sum\limits_{k=1}^{n} x^{t-1}_k \sum_{\ell\geq 2} \frac{(\eta \cdot e_k^\top  R \hat{y}^t)^\ell}{\ell!} \Big)\\
&\leq \ln\Big(1 + \eta \cdot (x^{t-1})^\top  R\hat{y}^t + 2\eta^2 \Big). \end{align*}

For the above we used the fact that $\sum_{\ell\geq 2} \frac{(\eta \cdot e_k^\top  R \hat{y}^t)^\ell}{\ell!} \leq \frac{\eta^2}{1-\eta} \leq 2\eta^2$, since $\eta\leq 1/2$. By exploiting now the inequality $\ln{(x)}\leq x-1$, we finally obtain the bound
\[ \ln\Big(\sum\limits_{k=1}^{n} x^{t-1}_k\cdot e^{\eta \cdot e_k^\top  R \hat{y}^t}\Big) \leq  \eta \cdot (x^{t-1})^\top  R\hat{y}^t +2\eta^2.\]
By carrying out similar calculations for the second logarithmic term, we will also get that  
\[\ln\Big(\sum\limits_{k=1}^{n} y^{t-1}_k \cdot e^{-\eta \cdot  e_k^\top  R^\top  \hat{x}^t}\Big) \leq  -\eta \cdot (\hat{x}^{t})^\top  R{y}^{t-1} +2\eta^2.\]
\noindent This gives us:
\begin{gather*}
D_{KL}((x^*,y^*)||(x^{t-1},y^{t-1})) -D_{KL}((x^*,y^*)||(x^{t},y^{t})) \\
\geq \eta \cdot (x^*)^{T}R\hat{y}^t - \eta \cdot (y^{*})^T R^\top  \hat{x}^t - \eta \cdot (x^{t-1})^\top  R\hat{y}^t+\eta \cdot (\hat{x}^{t})^\top  R{y}^{t-1} -4\eta^2.
\end{gather*}
By rearranging the terms, we obtain that
\begin{align*}
\eta \cdot \left( (\hat{x}^{t})^\top  R{y}^{t-1} - (x^{t-1})^\top  R\hat{y}^t \right) &\leq  D_{KL} ((x^*,y^*)||(x^{t-1},y^{t-1}))- D_{KL} ((x^*,y^*)||(x^{t},y^{t})) + 4\eta^2 \\ &\hphantom{\le} - \eta \cdot (x^*)^\top  R\hat{y}^{t} + \eta \cdot (\hat{x}^{t})^\top  Ry^*.
\end{align*}

Note now that since $(x^*, y^*)$ is a Nash equilibrium, and we are in a zero-sum game, then we know that $(x^*)^\top  R\hat{y}^{t} \geq v$, where $v$ is the value of the game. Similarly, $(\hat{x}^{t})^\top  Ry^* \leq v$. 
Hence these terms cancel out in the above equation and the proof is complete.
\end{proof}

The previous lemma is crucial as it provides a way to correlate the duality gap with the KL divergence. In particular, the left hand side of the formula is a proxy quantity for the duality gap, and converges to it should we choose a large enough $\xi$, as established via the following claim.

\begin{claim}\label{cl:duality}
    For any $t>0$, it holds that $\hat{x}^t$ is a $\frac{D_{KL}(\Tilde{x}^t || x^{t-1})}{\xi}$-best response against $y^{t-1}$, where $\Tilde{x}^t$ is the best response defined in \Cref{fact:best-response}. Similarly for the column player, $\hat{y}^t$ is a $\frac{D_{KL}(\Tilde{y}^t || y^{t-1})}{\xi}$-best response against $x^{t-1}$.
\end{claim} 
\begin{proof}
The proof uses the properties of the LogSumExp function, which we denote with $f$. In particular, given a mixed strategy $y$ of the column player, let $Ry$ be the payoff vector of the row player. Then $f$ is defined as $f(Ry) = log(\sum_{i=1}^n e^{(Ry)_i})$.
First of all, notice that the intermediate step of FLBR dynamics is essentially the softmax function, which is the gradient of $f$. Next, we use the fact that $f$ is the Fenchel conjugate of the KL divergence. Combining the two gives that if $x^{B}$ is a best response strategy against $y^{t-1}$, we have
\[
f(R y^{t-1}) = (\hat{x}^{t})^\top R y^{t-1} +\frac{1}{\xi} D_{KL}(x^B || x^{t-1}).
\]
Next, note that $f(R y^{t-1}) \ge \max_i e_i^\top R y^{t-1} = (x^B)^\top R y^{t-1}$. Substituting $f$ from above and rearranging the terms yields
\begin{equation}
\label{eq:x-approx-BR}
(x^B)^\top R y^{t-1}  - (\hat{x}^{t})^\top R y^{t-1}  \le \frac{1}{\xi} D_{KL}(x^B || x^{t-1}).
\end{equation}

To complete the proof, note that the argument applies to any best response against $y^{t-1}$, and hence we can use  $\Tilde{x}^{t}$ for $x^B$.

In a similar fashion, we can also establish that
\begin{equation}
\label{eq:y-approx-BR}
(x^{t-1})^\top R \hat{y}^{t} \leq  (x^{t-1})^\top R \Tilde{y}^{t}  + \frac{1}{\xi} D_{KL}(\Tilde{y}^t || y^{t-1}).
\end{equation}
\end{proof}

From the Claim we have the following:
\begin{corollary}\label[corollary]{cor:1}
Given a time step $T_0$, and any $t\le T_0$, $\xi  \ge \frac{T_0}{2\ln n} + \frac{1}{2\eta}$, and $\eta \le 1/2$, it holds that
\[V(x^{t-1}, y^{t-1}) \leq \frac{D^{t-1} - D^t }{\eta}+ 4(\ln n + 1)\eta.
\]
\end{corollary} 
\begin{proof}
We want to invoke \Cref{cl:duality}. However, to obtain concrete rates, we must bound the KL term. To that extent, the first step is to bound $D_{KL}(\Tilde{x}^t || x^{t-1})$ and $D_{KL}(\Tilde{y}^t || y^{t-1})$.  It is trivial to observe that  
$$D_{KL}(\Tilde{x}^t || x^{t-1}) \leq \ln\Big(\frac{1}{\min_{i:\Tilde{x}^t_i>0} x_i^{t-1}}\Big) \leq \ln\Big(\frac{1}{\min_{i\in [n]} x_i^{t-1}}\Big).$$ 

Next, since the payoffs are scaled in $[0,1]$ we can prove by induction that $x_i^t \ge x^0_i e^{-\eta t} = \frac{1}{n} e^{-\eta t}$ due to \Cref{assumption:1} (the same lower bound is also true for the column player, for $y_i^t$). 
Thus, our crude bound for the divergence we are interested in is that $D_{KL}(x^* || x^{t-1}) \le \ln n + \eta t$ and the same holds also for $D_{KL}(\Tilde{y}^t || y^{t-1})$. 
Now, by using \Cref{eq:x-approx-BR} and \Cref{eq:y-approx-BR} from \Cref{cl:duality}, we obtain

\[(\hat{x}^{t})^\top  R{y}^{t-1} - (x^{t-1})^\top  R\hat{y}^t \ge  (\Tilde{x}^t)^\top R y^{t-1} - (x^{t-1})^\top R \Tilde{y}^{t} - \frac{1}{\xi} (D_{KL}(\Tilde{x}^t || x^{t-1}) + D_{KL}(\Tilde{y}^t || y^{t-1})).
\]

Notice that by definition, $V(x^{t-1}, y^{t-1}) = (\Tilde{x}^t)^\top R y^{t-1} - (x^{t-1})^\top R \Tilde{y}^{t}$. Therefore, multiplying both sides with $\eta$ and using \Cref{lem:V-KL} gives us
\begin{align*}
\eta\left[V(x^{t-1}, y^{t-1}) - \frac{2(\ln n + \eta t)}{\xi}\right] &\le D^{t-1} - D^t  + 4\eta^2 \implies \\ V(x^{t-1}, y^{t-1}) &\le \frac{D^{t-1} - D^t}{\eta} + 4\eta + \frac{2(\ln n + \eta t)}{\xi}.
\end{align*}
For $\xi  \ge \frac{T_0}{2\ln n} + \frac{1}{2\eta}$, we bound the last term by $4 \eta \ln n$ and the proof is completed. 
\end{proof}

\begin{theorem}\label{th:flbrconvergence}
Suppose \Cref{assumption:1} holds, and fix any constant $\eta \le 1/2$. 
There exists a large enough $\xi$ so that
the rate of convergence for the KL divergence in order to reach an $O(\eta)$-NE is of the form $O(c^t)$ where $c<1$ is independent of t and dependent only on $\eta$ and the size of the game n. Similarly the convergence rate of the duality gap to reach and $O(\eta)$-NE is of the form $O(\frac{1}{\eta}\cdot c^t)$.
\end{theorem}
\begin{proof}
We will first assume \Cref{cor:1} for some value of $\xi$ as a function of some $T_0$, as dictated by \Cref{cor:1}, and once we have established the theorem,  we will drop $T_0$ from the analysis and find an appropriate lower bound for $\xi$. Assume that we have not reached a $(4\ln n + 6)\eta$ equilibrium. In other words, it holds that $V(x^{t-1}, y^{t-1}) \ge (4\ln n + 6)\eta$. Plugging this into \Cref{cor:1} gives us, after rearranging the terms:\[
    D^t \leq  D^{t-1} - 2\eta^2 = D^{t-1} \Big(1 - \frac{2\eta^2}{D^{t-1} }\Big).
\]

By the above, the KL divergence only decreases, and hence we have that $D^{t-1}\leq D^0 \leq 2\ln(n)$ (due to \Cref{assumption:1}). Thus, we deduce
\[ D^t  \leq   D^{t-1} \Big(1 - \frac{\eta^2}{\ln(n)}\Big).\]
Since $\eta \le 1/2 < \sqrt{\ln(n)}$, we can unroll the above inequality for all time steps up to $t$ to obtain
\[  D^t \leq D^0 \Big(1 - \frac{\eta^2}{\ln(n)}\Big)^t \le 2\ln(n)\Big(1 - \frac{\eta^2}{\ln(n)}\Big)^t.\]
This means that the KL divergence at time $t$ is bounded by $2\ln(n) \cdot c^t$, where $c<1$ is independent of $t$ and dependent on $\eta$ and $n$.
Coming now to the duality gap, we conclude by \Cref{cor:1} that 
\begin{equation}\label{eq:duality-bound}
     V(x^t, y^t) \leq  \frac{D^t}{\eta}+ 4(\ln n + 1)\eta  \leq  \frac{2\ln(n)}{\eta}\Big(1 - \frac{\eta^2}{\ln(n)}\Big)^t + 4(\ln n + 1)\eta.
\end{equation}
This upper bound combined with $V(x^t, y^t) \ge (4\ln n + 6)\eta$ implies that for any time step $t$, until we reach an approximate equilibrium, we have that $\eta \leq \frac{\ln(n)}{\eta}\Big(1 - \frac{\eta^2}{\ln(n)}\Big)^t $. By plugging this back into \Cref{eq:duality-bound}, we eventually get:
\[
    V(x^t, y^t) \leq \ln n \frac{4\ln n + 6}{\eta}\Big(1 - \frac{\eta^2}{\ln(n)}\Big)^t.
\]
To complete the proof, we must identify $\xi$, or in other words, drop $T_0$ from the analysis. To that end, it suffices to find the number of iterations actually needed to reach the claimed equilibrium, and set this as the value of $\xi$. We calculate
\begin{align*}
    \ln n\frac{4\ln n + 6}{\eta}\Big(1 - \frac{\eta^2}{\ln(n)}\Big)^{T_0} &\le (4\ln n + 6)\eta \implies 
        T_0 \ge \Bigg\lceil \frac{2\ln \eta - \ln(\ln n)}{\ln \Big(1 - \frac{\eta^2}{\ln(n)}\Big)} \Bigg\rceil.
\end{align*}
Since $T_0$ depends only on $\eta$, which is a fixed constant, and $n$ we can obtain a sufficiently large $\xi$, in accordance with \Cref{cor:1}.
\end{proof}
By the proof of the previous Theorem we can already derive a last-iterate convergence rate.
\begin{corollary}\label[corollary]{cor:new}
For $\epsilon > 0$, and $\eta \leq  \frac{\epsilon}{4\ln n + 6}$, the FLBR-MWU dynamics reach an $\epsilon$-approximate Nash equilibrium after $O\left(\frac{1}{\epsilon^2} \log\frac{1}{\epsilon}\right)$ iterations, where the hidden constant depends only on the dimension $n$ and not on the payoff matrix.
\end{corollary}
\begin{proof}
    From the previous analysis, to reach an $\epsilon$-NE it suffices to set $\eta =  \frac{\epsilon}{4\ln n + 6}$. Then, we can upper bound $T_0$ as follows:
    \begin{align*}
        T_0 &= \Big\lceil \frac{2\ln \eta - \ln(\ln n)}{\ln \Big(1 - \frac{\eta^2}{\ln(n)}\Big)} \Big\rceil 
        \le \frac{2\ln \eta }{\ln \Big(1 - \frac{\eta^2}{\ln(n)}\Big)}\\
        &\le \frac{2\ln\eta \ln n}{\eta^2}.
    \end{align*}
    Then, by substituting $\eta = \frac{\epsilon}{4\ln n + 6}$ and for $\xi = \Omega\Big(\frac{\ln \epsilon}{\epsilon^2}\Big)$ we conclude.
\end{proof}

\begin{remark}
\label{rem:convergence_clarification}
It may appear at first sight that \Cref{th:flbrconvergence} and \Cref{cor:new} are conflicting, since the former one refers to a geometric rate, whereas the latter yields $poly(1/\epsilon)$ iterations. The subtlety is that \Cref{th:flbrconvergence} states a geometric rate for a constant degree of approximation (to reach an $O(\eta)$-equilibrium), which translates to a constant time horizon till this is achieved, while \Cref{cor:new} establishes a rate for any $\epsilon>0$. We selected to state our results in this way, since in the following section we will prove a second convergence rate that holds after some initial phase. Thus, for a suitable parameter selection we can piece together the results to have an overall geometric rate, albeit with game-dependent constants.
\end{remark}

\subsection{Convergence to an exact equilibrium under uniqueness}
We proceed here to analyze the convergence until the method reaches an exact equilibrium. The technique here is based on a spectral analysis, and for this, we will need to further assume that the game has a unique Nash equilibrium $(x^*, y^*)$. This is a rather common assumption in many related works, and we do not view this as a severe restriction, since the set of zero-sum games with non-unique NE has Lebesgue measure equal to zero \citep{van1991stability}.

Let $t_0$ be the time at which we reach the approximate equilibrium described in \Cref{subsec:approx} and let $(x^{t_0}, y^{t_0})$ be the corresponding strategy profile. 
The first step in the remaining analysis is to establish that this approximate equilibrium can be close to the actual Nash equilibrium. This is ensured if $\eta$ is sufficiently small.

\begin{corollary}[implied by Theorem 3 in \citet{FMPV22}]
\label[corollary]{cor:delta-close}
For any $\delta>0$, and for any $q\geq 1$, there exists a sufficiently small $\eta$, such that $||(x^*, y^*) - (x^{t_0}, y^{t_0})||_q \leq \delta. $
\end{corollary}
Using the above, the asymptotic last-iterate convergence of FLBR (but without a rate) was established in \citet{FMPV22} by proving that the maximum eigenvalue of the Jacobian matrix at $(x^*, y^*)$ is strictly less than 1.
In order to obtain a rate of convergence, we give a more refined analysis, based on an idea utilized in \citet{NTHW21} (namely within the proof of their Theorem 5) for a fundamental problem in information theory.\footnote{In particular, the problem tackled by \citet{NTHW21} was  the convergence analysis of the Arimoto-Blahut algorithm for computing the Shannon's capacity of a discrete memoryless channel.}
\begin{theorem}
\label{theorem: geometric rate with uniqueness}
Let $(R,-R)$ be a zero-sum game with a unique NE $(x^*,y^*)$. For a sufficiently small $\eta$ and large enough $\xi$, such that $\eta\xi < 1$, the rate of convergence of the duality gap to the NE is geometric for the FLBR dynamics, in the form $A/b^t$, where $A$ and $b$ are determined by the norm of the Jacobian matrix evaluated at $(x^*,y^*)$.
\end{theorem}

\begin{remark}\label[remark]{rem:xi}
In order for \Cref{th:flbrconvergence} to hold, we needed $\xi$ to be larger than a quantity that depends on $\eta$. For \Cref{theorem: geometric rate with uniqueness} we need $\eta \xi < 1$. Unfortunately, both theorems cannot be guaranteed simultaneously with a single value of $\xi$. To verify this, we can lower bound $T_0$ by $\frac{1}{\eta^2}$ via the inequality $\ln(1-x) \le -x$. Hence, $\eta\xi > \frac{1}{\eta}$ which is unbounded. To remedy this issue, we propose a slight variant of the algorithm where we change the value of $\xi$ once we reach an approximate equilibrium, as guaranteed by \Cref{th:flbrconvergence}, so that from that point onwards, the inequality $\eta\xi < 1$ is satisfied. 
\end{remark}

The remainder of the section is dedicated to proving \Cref{theorem: geometric rate with uniqueness}.
First, we recall some basic facts established in \citet{FMPV22} that we use here, and for which uniqueness of equilibrium was needed. FLBR can be easily described as a discrete dynamical system, $\varphi(x,y) = \left(\varphi_1(x,y),\varphi_2(x,y) \right)$, such that $\varphi(x^t,y^t) = (x^{t+1},y^{t+1})$, and where $\varphi_{1,i}(x,y)$ is the $i$-th coordinate of $\varphi_1(x,y)$ and similarly for $\varphi_{2,i}(x,y)$, for any $i \in [n]$. The Jacobian of this system is a $2n\times 2n$ matrix, determined by the partial derivatives of $\varphi$.
Furthermore, when there exists a unique NE and $\eta\xi<1$, \citet{FMPV22} proved that there exists some $q\geq 1$, such that 
\[
    \lambda_{\max}\leq ||J(x^*,y^*)||_q  <1,
\] 
where $\lambda_{\max}$ is the maximum eigenvalue of the Jacobian matrix at the profile $(x^*,y^*)$.

For any $t\geq 0$, consider the strategy profile $(x(p),y(p)) = (1-p)\cdot (x^*,y^*) +p \cdot (x^t,y^t)$, with $ p \in (0,1)$, as a convex combination of the equilibrium and the profile $(x^t,y^t)$. In our proof, we will eventually need to argue about the Jacobian matrix at such convex combinations.

\begin{lemma}\label[lemma]{lem:mvt}
For $t \ge t_0$:  $|| (x^{t+1},y^{t+1})-(x^*,y^*)||_q \le ||(x^t,y^t)-(x^*,y^*)||_q\cdot ||J(x(p^t),y(p^t))||_q$.
\end{lemma}
First we show the following claim that we use in the proof of the above lemma. 
\begin{claim}\label{cl:delta_phi}
$\displaystyle\frac{d\varphi(x(p),y(p))}{dp} =  J(x(p),y(p)) \cdot \Big( x^t -x^* , y^t- y^*  \Big).$
\end{claim}

In the equation above, the term $(x^t -x^* , y^t -y^*)$ is a vector of $2n$ coordinates, where for each $i\in [n]$ the $i$-th coordinate equals $x^t_i - x^*_i$, and the $(n+i)$-th coordinate equals $y^t_i - y^*_i$.

\begin{proof}
For the row player, we have that for any $i$,
\begin{align*}
\frac{d\varphi_{1,i}(x(p),y(p))}{dp} &=
\sum_k \frac{dx_k(p)}{dp} \cdot \frac{d\varphi_{1,i}(x(p),y(p))}{dx_k(p)} + \sum_\ell \frac{dy_\ell(p)}{dp} \cdot \frac{d\varphi_{1,i}(x(p),y(p))}{dy_\ell(p)}\\
&= \sum_k \Big(   x^t_k -x^*_k \Big) \cdot J(x(p),y(p))_{ik}+\sum_\ell \Big( y^t_\ell -y^*_\ell  \Big) \cdot J(x(p),y(p))_{i,n+\ell}.
\end{align*}
The above holds because $\frac{dx_k(p)}{dp} = x^t_k -x^*_k$ and $\frac{dy_\ell(p)}{dp} = y^t_\ell -y^*_\ell$. Analogous expressions hold for  $\varphi_2$ as well, thus we conclude that
\begin{equation*}
\frac{d\varphi(x(p),y(p))}{dp} =  J(x(p),y(p)) \cdot \Big( x^t -x^* , y^t- y^*  \Big).\tag*{\qedhere}
\end{equation*}
\end{proof}

\begin{proof}[Proof of \Cref{lem:mvt}]
By the Mean Value Theorem (applied for our function $f^t = \varphi(x(p),y(p)):\mathbb{R}\to \mathbb{R}^{2n}$), for each time $t$, there is a $p^t \in (0,1)$ s.t.
\begin{align*}
|| (x^{t+1},y^{t+1})-(x^*,y^*)||_q &= || \Big(\varphi_1(x^{t},y^{t}),\varphi_2(x^t,y^t) \Big) - \Big(\varphi_1(x^*,y^*),\varphi_2(x^*,y^*) \Big) ||_q \\
&= ||f^t(1)-f^t(0)||_q  \\
&\leq ||\frac{df^t(p)}{dp}|_{p = p^t}||_q \cdot (1-0)\\
&= ||\Big((x^t,y^t)-(x^*,y^*)\Big)\cdot J(x(p^t),y(p^t))||_q  \\
&\leq ||(x^t,y^t)-(x^*,y^*)||_q\cdot ||J(x(p^t),y(p^t))||_q. 
\end{align*}
where the second inequality holds by the properties of the $q$-norm.
\end{proof}

With the above lemma and the continuity of the norm, we prove by induction the following:

\begin{lemma}\label[lemma]{lem:J-close}
    Given $\varepsilon>0$, there exists a sufficiently small $\delta>0$, such that if $|| (x^{t_0},y^{t_0})-(x^*,y^*)||_q \leq \delta$, then for any $t\geq t_0$. 
\end{lemma}
\begin{proof}
For the basis of the induction, consider $t = t_0$.
Regarding the Jacobian, first note that 
\begin{align*} ||(x(p^{t_0}),y(p^{t_0})) - (x^*,y^*) ||_q 
& = ||(1-p^{t_0}) (x^*,y^*) + p^{t_0}(x^{t_0},y^{t_0}) - (x^*,y^*) ||_q \\
& = ||p^{t_0}(x^{t_0},y^{t_0}) - p^{t_0}(x^*,y^*) ||_q \\
& \leq ||(x^{t_0},y^{t_0}) - (x^*,y^*) ||_q.
\end{align*}

Furthermore, by the continuity of the norm, for the given $\varepsilon$, there exists $\delta>0$ s.t. 
if $||(x^*,y^*) -  (x(p^{t_0}),y(p^{t_0}))||_q<\delta$, then $\Big|||J(x(p^{t_0}),y(p^{t_0}))||_q - ||J(x^*,y^*)||_q \Big| < \varepsilon$. 
Therefore, if we use this value of $\delta$, we get that if $|| (x^{t_0},y^{t_0})-(x^*,y^*)||_q \leq \delta$, then
$||(x(p^{t_0}),y(p^{t_0})) - (x^*,y^*) ||_q < \delta$ (by the previous analysis), and consequently
$||J(x(p^{t_0}),y(p^{t_0}))||_q  < ||J(x^*,y^*)||_q + \varepsilon$. 
This establishes the basis.

For the induction step, assume that the condition holds for some $t\geq t_0$. We will establish it for $t+1$. 

Since we have assumed that $\varepsilon$ satisfies $||J(x^*,y^*)||_q + \varepsilon<1$, the induction hypothesis yields that $||J(x(p^{t}),y(p^{t}))||_q < 1$.
Using this and \Cref{lem:mvt}, we get that $|| (x^{t+1},y^{t+1})-(x^*,y^*)||_q < || (x^{t},y^{t})-(x^*,y^*)||_q$. 
This also implies that if  $|| (x^{t_0},y^{t_0})-(x^*,y^*)||_q \leq \delta$,  this propagates throughout all the iterations for the same $\delta$, so that
$|| (x^{t+1},y^{t+1})-(x^*,y^*)||_q < \delta$.
And this in turn yields
\begin{align*} ||(x(p^{t+1}),y(p^{t+1})) - (x^*,y^*) ||_q 
& = ||(1-p^{t+1}) (x^*,y^*) + p^{t+1}(x^{t+1},y^{t+1}) - (x^*,y^*) ||_q \\
& = ||p^{t+1}(x^{t+1},y^{t+1}) - p^{t+1}(x^*,y^*) ||_q \\
& \leq ||(x^{t+1},y^{t+1}) - (x^*,y^*) ||_q\\
& < \delta.
\end{align*}

To finish the proof, we use the same argument as in the induction basis. Namely,  by the continuity of the norm, for the given $\varepsilon$, and for the $\delta$ that was identified in the induction basis, we will have that $\Big|||J(x(p^{t+1}),y(p^{t+1}))||_q - ||J(x^*,y^*)||_q \Big| < \varepsilon$, and thus 
\[||J(x(p^{t+1}),y(p^{t+1}))||_q  < ||J(x^*,y^*)||_q  + \varepsilon < 1.\qedhere \]
\end{proof}

We are now ready to complete the proof of our main theorem.
\begin{proof}[Proof of \Cref{theorem: geometric rate with uniqueness}]
Fix now a small $\varepsilon>0$ and let $\lambda = ||J(x^*,y^*)||_q + \varepsilon$ so that $\lambda<1$. 
By \Cref{lem:J-close} and applying repeatedly \Cref{lem:mvt}, we have that, for any $t\geq t_0$,$
|| (x^{t},y^{t})-(x^*,y^*)||_q < \lambda^{t-t_0} \cdot ||(x^{t_0},y^{t_0})-(x^*,y^*)||_q.$
Therefore, given $\varepsilon>0$, if we pick a sufficiently small $\eta$, we can ensure that there exists a small $\delta>0$, such that \Cref{cor:delta-close} holds with this $\delta$, i.e., $||(x^{t_0},y^{t_0})-(x^*,y^*)||_q< \delta$, and at the same time \Cref{lem:J-close} holds with the chosen $\varepsilon$ (and again for this $\delta)$.  
By the equivalence of the norms, all these yield that $|| (x^{t},y^{t})-(x^*,y^*)||_1 < K \cdot \delta \cdot \lambda^{t-t_0},$ for some integer $K>0$ independent of $t$, and dependent on $q$. This directly bounds the $L_1$ distances from the equilibrium strategies and, by applying \Cref{lem:norm to Duality gap}, we conclude that
\begin{equation}
    V(x^t, y^t) \leq 2 K \cdot \delta \cdot  \lambda^{t-t_0} + v - v = O(K \cdot \delta \cdot\lambda^t). \tag*{\qedhere}
\end{equation}
\end{proof}

As a final remark of this section, we would like to compare the convergence of FLBR (with the modification we just introduced as per \Cref{rem:xi})
against the convergence of OMWU, as described in \citet{Wei2021LinearLC}. The theoretical analysis of OMWU has a similar flavor with our work, in the sense that it also establishes convergence in two phases. Both works have a slow first phase and a much faster second one (and with game-dependent parameters for both methods). An advantage of our analysis is that it yields a simpler proof of convergence. At the same time, an advantage of the approach by \citet{Wei2021LinearLC} is that it determines a specific range of $\eta$ that works for any game, i.e., selecting any $\eta \le \frac{1}{8}$ is sufficient. Finally, in \Cref{sec:exper}, we also consider empirical comparisons, where FLBR has a much better performance. 
\section{Regret analysis}
\label{sec:regret}
In this section, we focus on yet another previously unexplored aspect of the FLBR method, beyond convergence. 

Most importantly, a fundamental question is whether FLBR is a no-regret algorithm. We provide a negative answer for this. 

So far, in the literature of methods with last-iterate convergence, there exist both no-regret algorithms (such as Optimistic MWU \citep{Daskalakis2019LastIterateCZ}) and algorithms with regret (such as Extra-Gradient). 
We note that the existence of regret by itself is not necessarily a negative indication for an algorithm's performance. For example, OMWU is outperformed by algorithms that have regret, as discussed in \citet{caiforget}.

\begin{theorem}\label{thm:regret}
    FLBR is not a no-regret algorithm when $\xi$ is sufficiently large.
\end{theorem}

We first restate the FLBR dynamics, so that each iteration is replaced by two steps. We do this so as to explicitly view FLBR within the framework of online learning algorithms with gradient feedback. Hence, in each step, each player observes the payoff of her pure strategies\footnote{Note that this is precisely the gradient information, since e.g. $\frac{\partial (x^t)^\top  Ry^t}{\partial x_i} = e_i^\top Ry^t$.} and updates the mixed strategy accordingly.
This gives the following formulation for the row player (and analogously for the column player). For technical convenience, we assume the initial profile is indexed as $(x^{-1}, y^{-1})$:
    \begin{gather*}
         x_i^{2t} = x_i^{2t-1}\cdot \frac{e^{\xi \cdot e_i^\top  R{y}^{2t-1} }}{\sum\limits_{j} x_j^{2t-1}\cdot e^{\xi \cdot e_j^\top  R{y}^{2t-1}}},\
         x_i^{2t+1} = x_i^{2t-1}\cdot \frac{e^{\eta \cdot e_i^\top  R{y}^{2t} }}{\sum\limits_{j} x_j^{2t-1}\cdot e^{\eta \cdot e_j^\top  R{y}^{2t}}}, \quad t\ge 0. \label{eq1}
    \end{gather*}
     The example that we use for proving the theorem is the simple Matching Pennies game: 
    \[R = \begin{bmatrix} +1 & -1\\ -1 & +1\end{bmatrix}.\]
Note that we have assumed the games we study are in $[0, 1]$. We could bring the Matching Pennies game in this form as well, but for ease of presentation, we present the analysis here for the original form of the game. 
We use the initialization $x^{-1} = (1-\delta, \delta)$ and $y^{-1} = (\delta, 1 - \delta)$, for some small $\delta \in (0, 1/2)$. With this at hand, we can break down the proof of \Cref{thm:regret} into the lemmata that follow.
For simplicity, we will carry out the proof here assuming $\xi\to\infty$. Under this, note that by \Cref{fact:best-response},  $x^{0}$ is a best response to $y^{-1}$, and hence we get that $x^0 = (0, 1)$. In fact, we can inductively extend this argument.
\begin{claim}\label{cl:less than 1/2}
        For any $t\geq 0$, it holds that $x_1^{2t-1} > \frac{1}{2}$ and $y_1^{2t-1} < \frac{1}{2}$.
\end{claim}
Pairing this with \Cref{fact:best-response}, we get that $x^{2t}= (0, 1)$ (as a best response to $y^{2t-1}$, for any $t$) and symmetrically $y^{2t} = (0, 1)$. Now we are in position to explicitly compute $x_1^{2t-1}$.

\begin{lemma}\label[lemma]{lem:x_closed_form}
    For sufficiently large $\xi$ we get $\displaystyle x_1^{2t+1} = (1-\delta)[1 - \delta(1 - e^{2\eta(t+1)})]^{-1}$. 
\end{lemma}

\begin{proof}
Recall that 
       \begin{gather*}
        x_1^{2t+1} = x_1^{2t-1}\cdot \frac{e^{\eta \cdot e_1^\top  R{y}^{2t} }}{\sum\limits_{j} x_j^{2t-1}\cdot e^{\eta \cdot e_j^\top  R{y}^{2t}}} = x_1^{2t-1}\cdot \frac{e^{-\eta}}{\sum\limits_{j} x_j^{2t-1}\cdot e^{\eta \cdot e_j^\top  R{y}^{2t}}},\\
        x_2^{2t+1} = x_2^{2t-1}\cdot \frac{e^{\eta \cdot e_2^\top  R{y}^{2t} }}{\sum\limits_{j} x_j^{2t-1}\cdot e^{\eta \cdot e_j^\top  R{y}^{2t}}} = x_2^{2t-1} \cdot \frac{e^{\eta}}{\sum\limits_{j} x_j^{2t-1}\cdot e^{\eta \cdot e_j^\top  R{y}^{2t}}}. 
    \end{gather*}
    For brevity, let $x_1^{2t+1} = a^t$ and $x_2^{2t+1} = b^t$ we get that
    \begin{gather*}   
    a^t = a^{t-1} \cdot \frac{e^{-\eta}}{a^{t-1} e^{-\eta} + b^{t-1} e^{\eta}}, \\
    b^t = b^{t-1} \cdot \frac{e^{\eta}}{a^{t-1} e^{-\eta} + b^{t-1} e^{\eta}}.
    \end{gather*}
    Note that $a^t + b^t = 1$ so we get
    \begin{gather*}a^t = a^{t-1} \cdot \frac{e^{-\eta}}{a^{t-1} e^{-\eta} + (1-a^{t-1})e^{\eta}} = \frac{a^{t-1} e^{-\eta}}{a^{t-1}(e^{-\eta} - e^{\eta}) + e^{\eta}} \implies  \\
    \frac{1}{a^t} = 1 - e^{2\eta} + e^{2\eta}\frac{1}{a^{t-1}} \implies \\
    \frac{1}{a^t} - 1 = e^{2\eta} \left(\frac{1}{a^{t-1}} - 1\right) \implies\\
    \frac{1}{a^t} - 1 = e^{2\eta(t+1)}\left(\frac{1}{a^{-1}} - 1\right).
    \end{gather*}
    Recall that $a^{-1} = x_1^{-1} = 1-\delta$ so we get that
\begin{equation*}
    \frac{1}{a^t} = 1 + e^{2\eta(t+1)}\frac{\delta}{1 - \delta} \implies x_1^{2t+1} = \frac{1-\delta}{1 - \delta(1 - e^{2\eta(t+1)})}. \tag*{\qedhere}
\end{equation*}    
\end{proof}

Clearly, we also have $x_2^{2t+1} = 1 - x_1^{2t+1}$.   
Due to symmetry we obtain that $ y_2^{2t+1} = x_1^{2t+1}$; thus, we have obtained a closed form for the dynamics. 
The proof is then completed by the next lemma.

\begin{lemma}\label[lemma]{lem:regret_done}
For  sufficiently small $\delta$ and sufficiently large $\xi$, the regret of the algorithm for the row player against the fixed strategy $x = (0,1)$, up until time $T$ is $\Omega(T)$.  

\begin{proof}
For a given $T$, we compute the total payoff of the row player for the first $2T$ iterations when both players use FLBR. Since at the even steps of this process, the strategy of both players is $(0, 1)$, we get:

 \begin{align*}
\sum_{t=0}^{2T-1} {x^t}^\top  R y^t &= T \cdot (0,1)^\top  R (0,1) + \sum_{t=0}^{T-1} {x^{2t+1}}^\top  R y^{2t+1} \\
&= T + \sum_{t=0}^{T-1}(a^t, 1-a^t)^\top  R (1-a^t, a^t) \\
&= T + \sum_{t=0}^{T-1}(a^t, 1-a^t)^\top  (1 -2a^t, -1 + 2a^t) \\
&= T + \sum_{t=0}^{T-1}a^t - 2(a^t)^2 - 1 + 2a^t + a^t - 2(a^t)^2 \\
&= \sum_{t=0}^{T-1} 4a^t(1  - a^t),
\end{align*}
where once again we set $x_1^{2t+1} = a^t$. 

Next, we compute the payoff of the fixed strategy $x^* = (0,1)$ for the row player, against the column player playing in each iteration the FLBR strategy $y^t$ as computed by the previous analysis. This is equal to: 
\begin{align*}
\sum_{t=0}^{2T-1} {x^t}^\top  R y^t &= T \cdot (0,1)^\top  R (0,1) + \sum_{t=0}^{T-1} (0,1)^\top  R y^{2t+1} \\
&= T + \sum_{t=0}^{T-1}(0,1)^\top  (1 -2a^t, -1 + 2a^t) \\
&= \sum_{t=0}^{T-1} 2a^t.
\end{align*}
Hence, the regret for the row player when choosing her FLBR strategy against $x^*$ is
\[
    \text{Reg}_{\text{FLBR}} \ge \sum_{t=0}^{T-1} 2a^t - \sum_{t=0}^{T-1} 4a^t(1  - a^t) = \sum_{t=0}^{T-1} 2a^t(2a^t - 1).
\]
To upper bound the expression we use that $a^T = 1/2$. Hence we have that 
\begin{align*}
    \frac{1 - \delta(1 - e^{2\eta(T+1)})}{1-\delta} &= 2  \implies\\
    \delta e^{2\eta(T+1)} &= 1-\delta  \implies\\
    2\eta(T+1) &= \ln\Big(\frac{1-\delta}{\delta}\Big). 
\end{align*}
Thus, up to time $\lceil\frac{T - 1}{2}\rceil$ we have that
\[
a^t \ge \frac{1-\delta}{1 - \delta\left(1 - \sqrt{\frac{1-\delta}{\delta}}\right)} = \frac{1-\delta}{1 - \delta + \sqrt{\delta - \delta^2}} .
\]
For $\delta \to 0$ the expression tends to 1 so there is a sufficiently small $\delta$
 such that $a^t \ge .95$ for $t \le \lceil\frac{T+1}{2}\rceil$. Piecing everything together we get that
 \begin{align*}
     \text{Reg}_{\text{FLBR}} &\ge  \sum_{t=0}^{T-1} 2a^t(2a^t - 1) \\
     &\ge  \sum_{t=0}^{\lceil\frac{T - 1}{2}\rceil} 2a^t(2a^t - 1). \\
     \text{Reg}_{\text{FLBR}} &\ge  0.855 \cdot T \ \text{over } 2T \text{ rounds,} 
 \end{align*}
 which completes the proof.   
\end{proof}
\end{lemma}

\section{Experimental evaluation}\label{sec:exper}

Experimentally, the method had already seemed to be promising in \cite{FMPV22}, where it was compared against the Optimistic Multiplicative Weights Update (OMWU) method. However, given that OMWU does not perform so well in practice, here we provide comparisons of FLBR against Optimistic Gradient Descent-Ascent (OGDA), which is among  the fastest and state of the art last-iterate methods for bilinear games \citep{daskalakis2018training}.

We have performed 2 main types of comparisons. One class of experiments concerns random games, and more specifically payoff matrices that are populated from a standard Gaussian distribution. We have also sought additional games that are more structured and far from random. To that end, we have used the generalized Rock Paper Scissors (RPS) game of higher dimensions, and also some classes of other symmetric zero-sum games. 
In addition to these, we also consider a game discussed in \Cref{subsec:forget} that has to do with the notion of forgetfulness. 
In our experiments we have used both the completely uniform initialization so as to be consistent with our theory (\Cref{assumption:1}), but we have also experimented other initializations. E.g. for the RPS games, where the uniform profile is an exact equilibrium, we have started the method from a random mixed strategy profile.

\subsection{Implementation aspects and variants of FLBR} In the initial paper of \cite{FMPV22}, the method was tested for some values of $\eta$ and $\xi$, but still keeping $\xi$ to the same value throughout all the iterations. We also provide such experiments, but given \Cref{rem:xi}, in order to be more consistent with our theoretical analysis, we have additionally experimented with using different values of $\xi$ in the 2 phases of the trajectory (i.e., a first value for the initial convergence to an approximate equilibrium and a different one for the convergence to the exact solution).
To begin with, a simple first idea is to use essentially an infinite $\xi$, which according to \Cref{fact:best-response} translates to picking a best response directly in the intermediate step. According to our theory, this should lead to a fast initial convergence to an approximate equilibrium. However, this may be followed by a failure to contract inside the equilibrium region, given that $\eta\xi > 1$. We observe this phenomenon in the classical Rock-Paper-Scissors game (\Cref{fig:rps}) where we compare FLBR with OGDA, running both algorithms with $\eta = 0.1$. We can see there that this variant of FLBR gets stuck at an accuracy of 0.01 without making further progress.

\begin{figure}
    \centering
    \begin{minipage}{0.45\textwidth}
    \centering
    \includegraphics[width=0.8\linewidth]{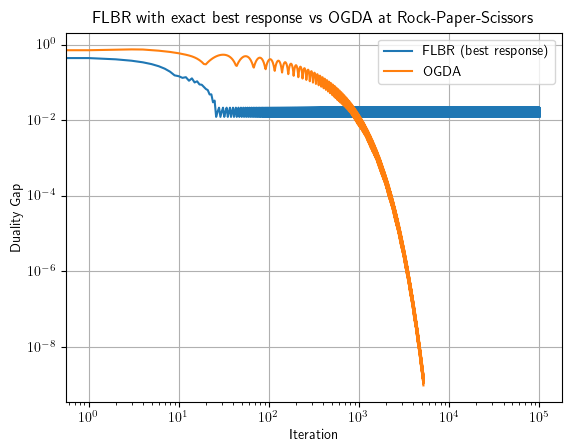}
    \caption{FLBR gets stuck when picking a best response in the intermediate step.}
    \label{fig:rps}
    \end{minipage}\hfill
    \begin{minipage}{0.45\textwidth}
            \centering
    \includegraphics[width=0.8\linewidth]{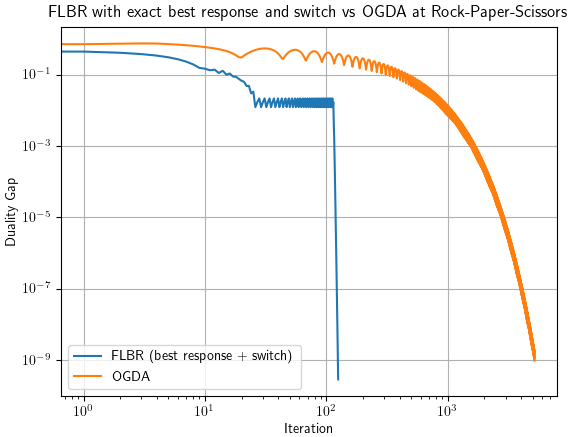}
    \caption{Convergence acceleration when switching to a finite $\xi$.}
    \label{fig:rps2}
    \end{minipage}
\end{figure}

To fix the above issue convergence we must define a criterion to switch in the intermediate step, from a best response selection to some finite $\xi$. For this we used $\xi = 100$, based on the work of \cite{FMPV22}, and given also that we did not observe any improvements for other values of $\xi$. 
We therefore propose a simple heuristic that implements the switch: We start the method by using a best response strategy in the intermediate step, and if in the last 200 iterations the duality gap has not been reducing, we modify the intermediate step update by using $\xi=100$. In \Cref{fig:rps2}, we see that the algorithm progresses, and it is actually much faster than OGDA.

Clearly, this experiment was just for a single $3\times 3$ game, which is quite small. As a next step, we therefore experimented with the standard generalization of the RPS game to higher dimensions, referred to as generalized RPS games.
For these games, to have a fair comparison against OGDA, we also fine-tuned the learning rate of OGDA to $\eta = 0.01$ because it failed to converge with a larger step. Now, the image in \Cref{fig:rps101} is even clearer. We can see that the FLBR algorithm gets stuck during its first phase, but after the switch it starts making progress again, without jumping to the Nash equilibrium immediately. Most importantly, we also see that it outperforms OGDA.

\begin{figure}
    \centering
    \begin{minipage}{0.45\textwidth}
            \centering
            \includegraphics[width=0.8\linewidth]{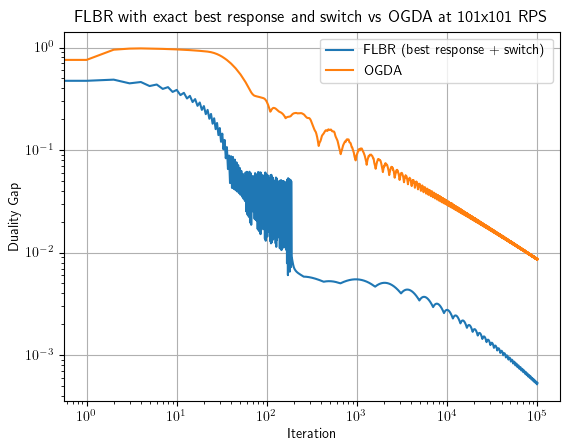}
    \caption{Comparisons for higher dimensional RPS games.}
    \label{fig:rps101}
    \end{minipage}
    \hfill
    \begin{minipage}{0.45\textwidth}
    \centering
    \includegraphics[width=0.8\linewidth]{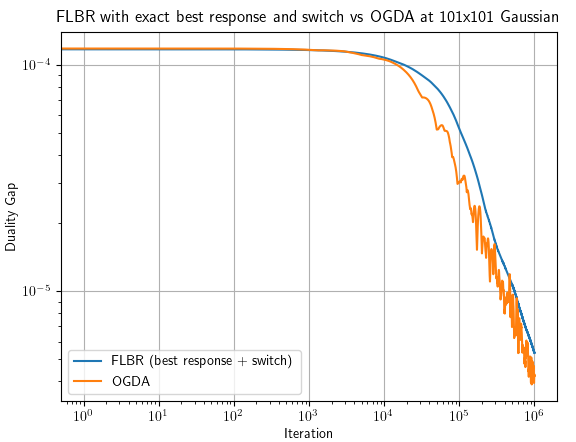}
    \caption{Comparisons with random games of size $101\times 101$.}
    \label{fig:flbr_gauss_final}
    \end{minipage}
\end{figure}

In order to move away from the more structured RPS games, we also considered random games, where we sample each entry of the payoff matrix independently from the standard Gaussian and normalize the entries to $[0,1]$.
The comparison between OGDA and FLBR is shown in \Cref{fig:flbr_gauss_final}, where we observe two things. First of all, that the algorithms are essentially matched in performance, and second, that the switch in FLBR never actually occurred. 

\subsection{Additional experiments}\label{app:exp}
Our additional experiments follow a similar line of thought as the ones presented above. In this subsection we focus on the simpler version of FLBR where we use the same value of $\xi$ throughout all iterations in contrast to the experiments in the first part of the experimental section, where we followed the theoretical analysis and \Cref{rem:xi}. We start with random Gaussian games, where OGDA has a slight advantage over this version of FLBR and then we present constructions of not so random games, with some inherent structure, which slow down OGDA but not FLBR.

\paragraph{Initializations} As stated in \Cref{assumption:1}, for the theoretical part of the paper we always initialize FLBR with the uniform distribution, i.e. $x_i =y_i = 1/n$. Here we deem useful to explore more options. Specifically, we test the following starting points:
\begin{itemize}
    \item Uniform distribution.
    \item Almost pure strategy profile: $x_1 = y_1 = 1 - 1/n$ and $x_i = y_i = \frac{1}{n(n-1)}$.
    \item Random: we sample $x, y$ from $U(0,1)$ and then rescale them.
    \item Sequential: $x_i = y_i = \frac{2i}{n(n+1)}$.
\end{itemize}

\paragraph{Assumptions on $\eta$, $\xi$} In the theoretical part of the paper, we did not need any major assumption for $\eta$ and $\xi$ (apart from $\xi$ being large enough) for reaching an approximate equilibrium. However, for the convergence to the exact solution, we needed to use $\eta\xi<1$, to prove \Cref{theorem: geometric rate with uniqueness}. In our experiments, we also tested combinations of values for these two parameters that violate this condition. What we observe experimentally is that the method can perform well even without this constraint (recall e.g., that in the first line of experiments, we also used $\xi=100$ and values of $\eta$ for which $\eta\xi>1$),  but certainly not for any arbitrary combination.

\subsubsection{Random games}
In \Cref{fig:Gaussian-50,fig:Gaussian-500} we see the comparisons between FLBR and OGDA for further Gaussian games of dimensions 50 and 500, where each entry of the payoff matrix is filled by sampling from the Gaussian distribution. What we observe is that OGDA performs better (as expected by the existing smoothed analysis for OGDA) and that FLBR is close but on average slower than OGDA.

\begin{figure}[H]
\centering
    \begin{minipage}{0.33\textwidth}
        \centering
        \includegraphics[width=\textwidth]{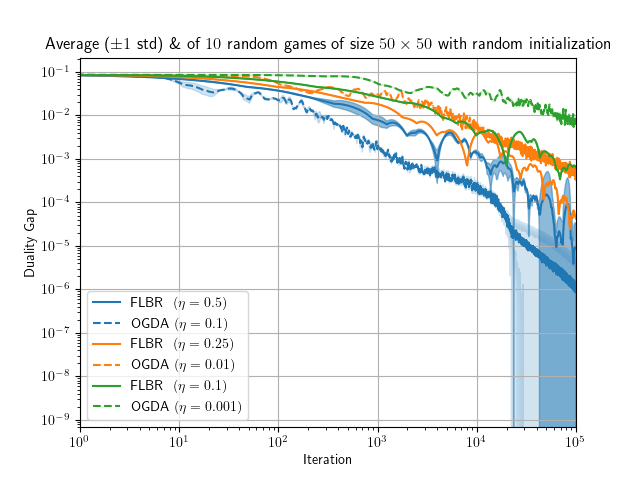}
        \label{fig:G50rng}
    \end{minipage}\hfill
    \begin{minipage}{0.33\textwidth}
        \centering
        \includegraphics[width=\textwidth]{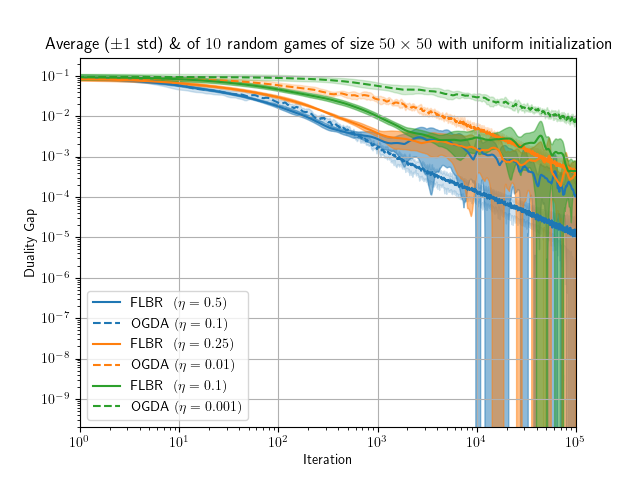}
        \label{fig:G50uni}
    \end{minipage}\hfill
        \begin{minipage}{0.33\textwidth}
        \centering
        \includegraphics[width=\textwidth]{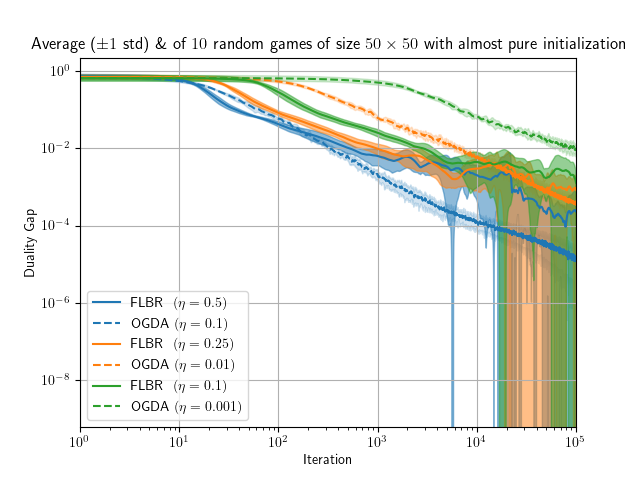}
        \label{fig:G50one}
    \end{minipage}
    \caption{Random Gaussian $50\times 50$ games with various initializations.}
    \label{fig:Gaussian-50}
\end{figure}
\begin{figure}[H]
\centering
    \begin{minipage}{0.33\textwidth}
        \centering
        \includegraphics[width=\textwidth]{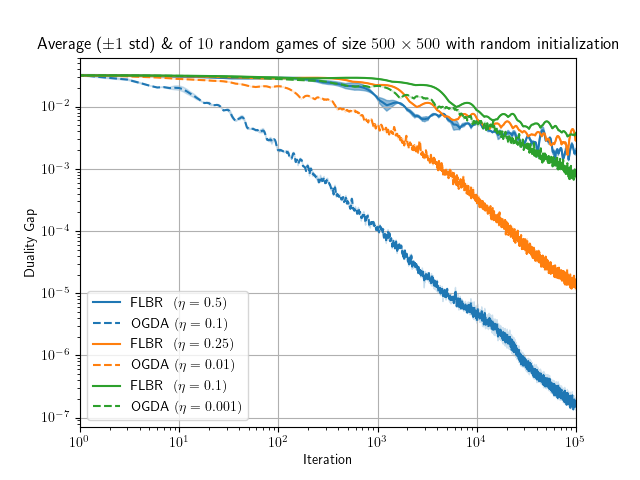}
        \label{fig:G500rng}
    \end{minipage}\hfill
    \begin{minipage}{0.33\textwidth}
        \centering
        \includegraphics[width=\textwidth]{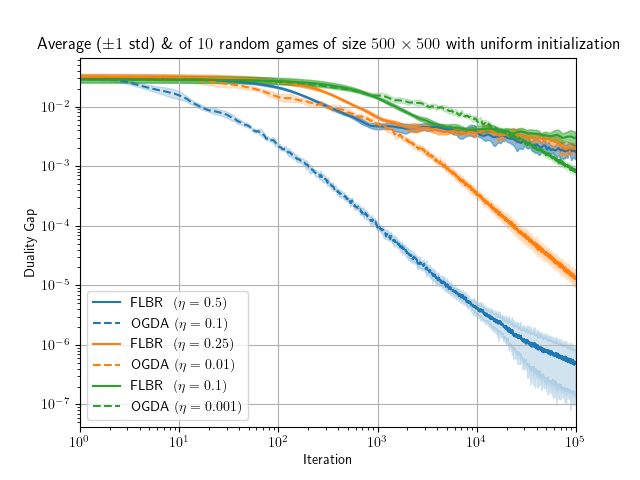}
        \label{fig:G500uni}
    \end{minipage}\hfill
        \begin{minipage}{0.33\textwidth}
        \centering
        \includegraphics[width=\textwidth]{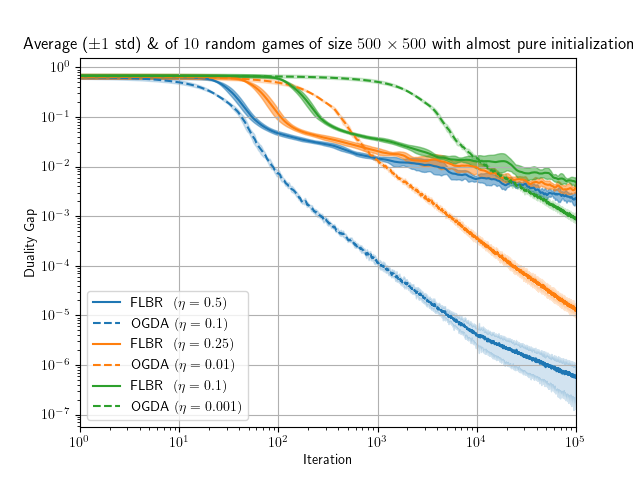}
        \label{fig:G500one}
    \end{minipage}
     \caption{Random Gaussian $500\times 500$ games with various initializations.}
      \label{fig:Gaussian-500}
\end{figure}

\subsubsection{Structured games}

We have already presented our results on the Generalized Rock-Papers-Scissors game, which is arguably among the most famous zero-sum games. Here we also present comparisons using two more classes of more structured games.

First, we performed comparisons for games where the payoff matrix $R$ is of low rank. Such games differ from random games, where with high probability the matrix has full rank. 
We constructed matrices, where the rank is approximately 5-10\% of the dimension. 

Interestingly, what we observe in \Cref{fig:low-rank-50,fig:low-rank-500}, is that FLBR is performing better than OGDA. The figures depict the comparisons for $50\times 50$ games where the rank is $5$ and for $500\times 500$ games with rank equal to 25. An additional observation is that FLBR seems more robust against the various initializations that were used. For example OGDA, under the random and the uniform initialization does not converge for some choices of $\eta$.

\begin{figure}[H]
\centering
    \begin{minipage}{0.33\textwidth}
        \centering
        \includegraphics[width=\textwidth]{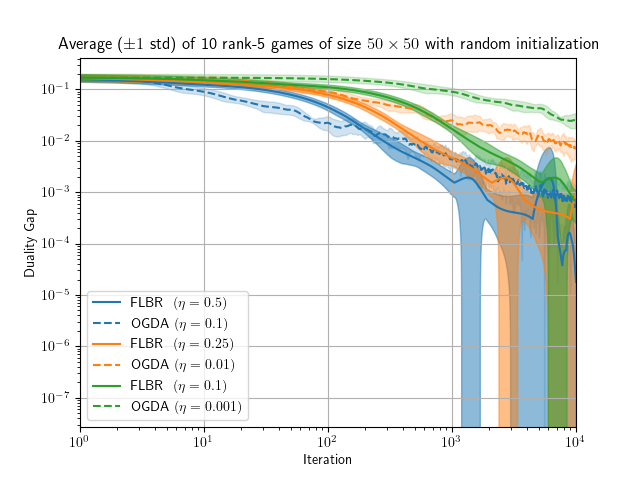}
    \end{minipage}\hfill
    \begin{minipage}{0.33\textwidth}
        \centering
        \includegraphics[width=\textwidth]{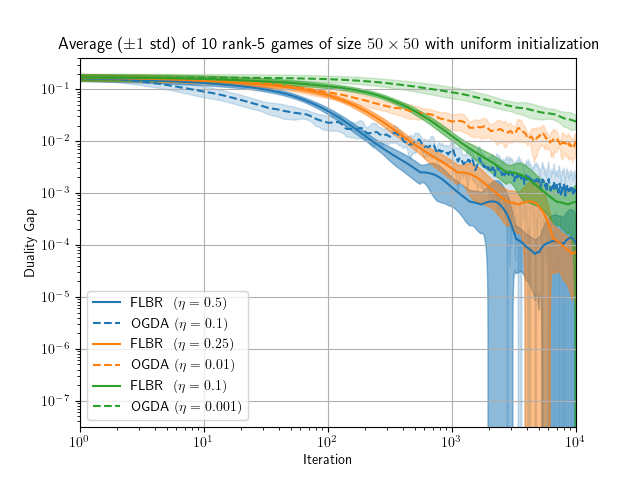}
    \end{minipage}\hfill
        \begin{minipage}{0.33\textwidth}
        \centering
        \includegraphics[width=\textwidth]{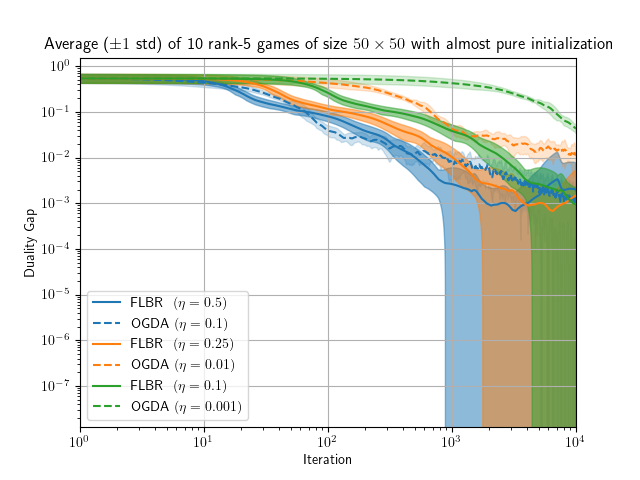}
    \end{minipage}
     \caption{Games with low rank payoff matrix of size $50\times 50$ with various initializations.}
     \label{fig:low-rank-50}
\end{figure}

\begin{figure}[H]
\centering
    \begin{minipage}{0.33\textwidth}
        \centering
        \includegraphics[width=\textwidth]{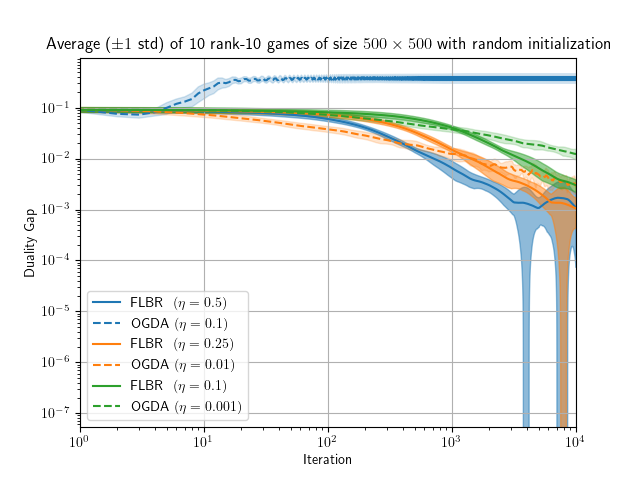}
    \end{minipage}\hfill
    \begin{minipage}{0.33\textwidth}
        \centering
        \includegraphics[width=\textwidth]{figures/extraexperiments/lowrank50_uni.png}
    \end{minipage}\hfill
        \begin{minipage}{0.33\textwidth}
        \centering
        \includegraphics[width=\textwidth]{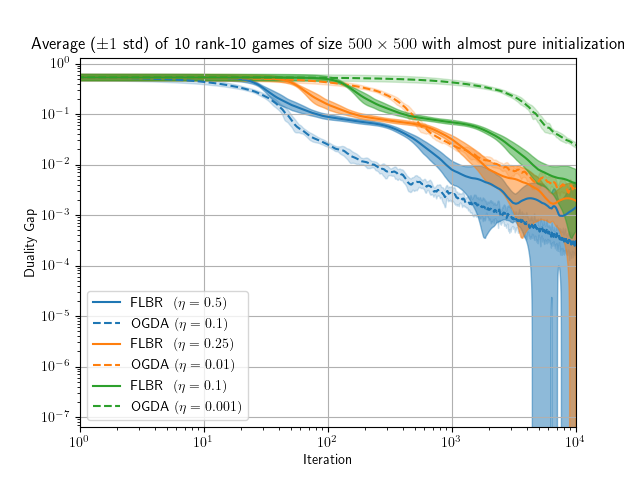}
    \end{minipage}
     \caption{Games with low rank payoff matrix of size $500\times 500$ with various initializations.}
         \label{fig:low-rank-500}
\end{figure}

Moving on,  we also tested a class of symmetric zero-sum games, which again is more structured than random games. In order to construct such families, we used the following formula for filling in the entries of the payoff matrix, where $P_{ij}^n$ is the entry of $P$ at $(i, j)$ when $P$ is $n\times n$. Here symmetry is enforced, given the dependence on $i+j$, 
\begin{equation}
\label{eq:P}
P^n_{ij} = \frac{1}{n} (i + j - 2) \text{ mod } n.    
\end{equation}

We note that for this class, we did not use the uniform initialization as this is an equilibrium of the game. What we observe in \Cref{fig:P50,fig:P500}, is that FLBR is having an advantage over OGDA for smaller dimensions, while OGDA becomes just slightly better, for the sequential and the almost pure initialization. The two methods have a very similar performance under the random initialization. Again, we observe a better robustness of FLBR with respect to the various initializations and the values of $\eta$. For example, we see that OGDA does not manage to converge for some of the choices used for $\eta$.

\begin{figure}[H]
\centering
    \begin{minipage}{0.33\textwidth}
        \centering
        \includegraphics[width=\textwidth]{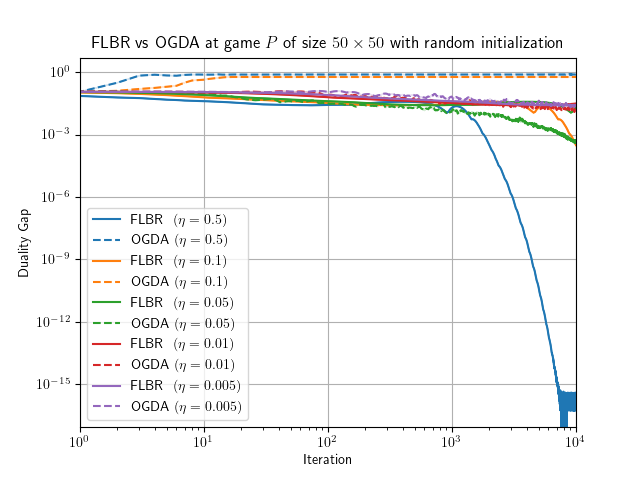}
    \end{minipage}\hfill
    \begin{minipage}{0.33\textwidth}
        \centering
        \includegraphics[width=\textwidth]{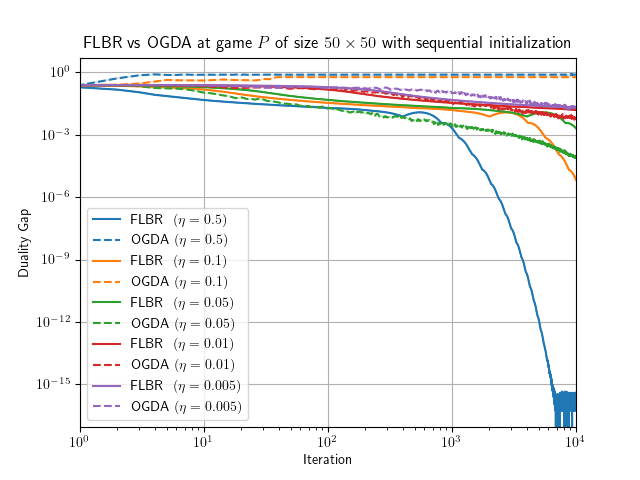}
    \end{minipage}\hfill
        \begin{minipage}{0.33\textwidth}
        \centering
        \includegraphics[width=\textwidth]{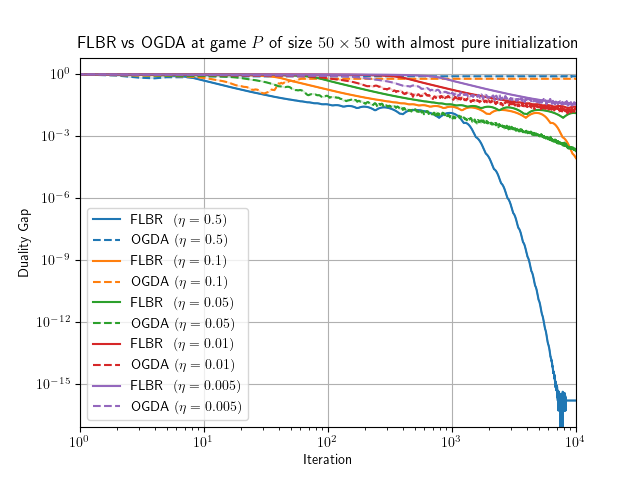}
    \end{minipage}
     \caption{Structured games defined by \Cref{eq:P}, of size $50\times 50$ with various initializations.}
     \label{fig:P50}
\end{figure}

\begin{figure}[H]
\centering
    \begin{minipage}{0.33\textwidth}
        \centering
        \includegraphics[width=\textwidth]{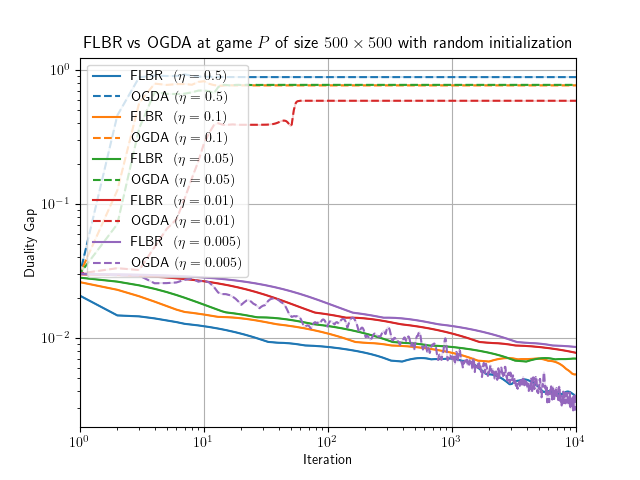}
    \end{minipage}\hfill
    \begin{minipage}{0.33\textwidth}
        \centering
        \includegraphics[width=\textwidth]{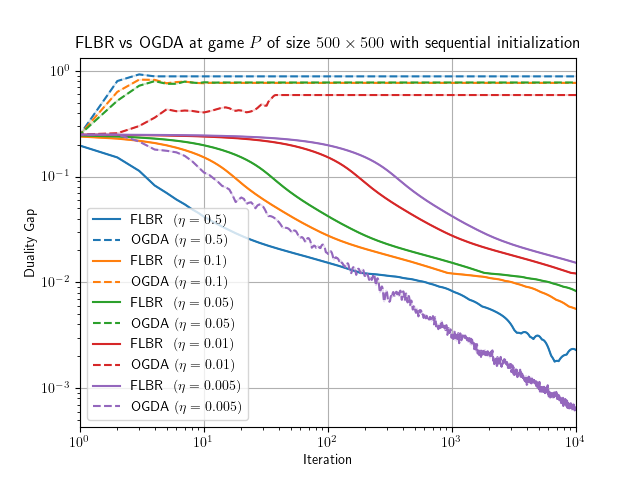}
    \end{minipage}\hfill
        \begin{minipage}{0.33\textwidth}
        \centering
        \includegraphics[width=\textwidth]{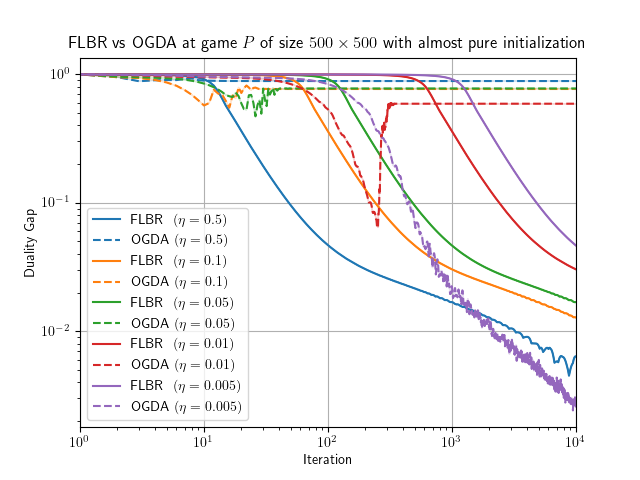}
    \end{minipage}
    \caption{Structured games defined by \Cref{eq:P}, of size $500\times 500$  with various initializations.}
    \label{fig:P500}
\end{figure}

\subsection{Experimental conclusions}

Our main findings and conclusions from all our experiments presented can be summarized as follows:
\begin{itemize}
    \item The FLBR method with the same $\xi$ throughout all iterations has a comparable, and sometimes better, performance against OGDA in more structured games. This is true for example both for generalized RPS games and for symmetric games. 
    \item For random Gaussian games, the methods are comparable up to a point, with OGDA being better both in the number of iterations needed and the time elapsed per game. Nevertheless, FLBR is still close enough. The good performance of OGDA in random games can be explained by \cite{anagno_smooth}, via last iterate analysis under the celebrated framework of \textit{smoothed analysis} \citep{ST04}.
    \item Most importantly, our heuristic with the switch in the values of $\xi$ seems more robust and attains the best of both worlds, in the sense that it matches the performance of OGDA in random games, and it outperforms it in more structured games. 
    \item Finally, apart from the number of iterations shown in the previous figures, we present some indicative time comparisons between FLBR and OGDA in \Cref{tab:timemeasurementsgaussian} for a random Gaussian game. Again the conclusion remains the same, that FLBR matches and even outperforms OGDA.  
\end{itemize}

\begin{table}[H]
\centering
\begin{minipage}{\linewidth}
\centering
\caption{Comparison in Gaussian games}
\label{tab:timemeasurementsgaussian}
\begin{tabular}{lcccc}
\toprule
& \multicolumn{4}{c}{Time (sec) to accuracy} \\
\cmidrule(lr){2-5}
     & $10^{-2}$  & $10^{-3}$ & $10^{-4}$ & $10^{-5}$ \\
\midrule
OGDA & $4.5 \cdot 10^{-5}$ & 0.86 & 19.22 & 168.35 \\
FLBR & $7.7 \cdot 10^{-5}$ & 0.77 & 3.81 & 33.12 \\
\bottomrule
\end{tabular}
\end{minipage}\hfill
\end{table}

Overall, even though the theoretical analysis of FLBR comes with the caveat of game-dependent parameters in its geometric convergence rate, the experiments reveal a competitive performance against OGDA. One more conclusion that arises from the experiments (see \Cref{fig:rps,fig:rps2}) is that FLBR seems to exhibit better robustness when varying $\eta$, unlike OGDA. We therefore conclude that the combination of different learning rate parameters, $\eta$ and $\xi$, in FLBR can be viewed as a promising direction that could motivate further future works.
\section{Discussion}
The goal of this work was to provide a more complete treatment, both theoretically and experimentally, of the FLBR method. In doing so, we also answered an open question from \cite{FMPV22}, and established concrete rates of convergence. We conclude our work with a discussion on the recently introduced notion of {\it forgetfulness}.

\subsection{Forgetfulness}
\label{subsec:forget}
In a very recent work, \cite{caiforget} provided further insights on the performance of OMWU and related dynamics, as compared to OGDA. Their work was motivated  by \cite{PPSC23}, where analogous intuitions were given for the fictitious play algorithm. 
In a nutshell, \cite{caiforget} attributed the cause of relatively slow convergence of OMWU to a notion they term ``forgetfulness". Although they did not provide a formal definition, intuitively, a method that is not forgetful allows the produced strategies to get stuck in almost the same profile over many iterations, which slows down convergence.

They designed a specific $2\times 2$ game and showed that this slow behavior can occur under OMWU, whereas OGDA does not exhibit the same issues. In \Cref{app:forget}, we extend their experiment and compare FLBR with OGDA in the same game. The main takeaway is that FLBR does not suffer from slow convergence. On the contrary, it even outperforms OGDA. We refer to \Cref{app:forget} for further details.

Overall, even though this was only one example, it conveys the intuition that the intermediate step of FLBR, using large $\xi$, has a particular effect in the dynamics: it makes the algorithm forgetful, and thus faster, albeit with the cost of adding regret, as shown in \Cref{sec:regret}. In general, we believe that the topic of forgetfulness deserves further exploration for learning dynamics.

\subsection{Future work}
As a step towards further explorations for the performance of FLBR, it would be interesting to study if our results generalize beyond bilinear payoffs to other classes of games. A first step would be to investigate (classes of ) convex-concave games. Other games of interest include potential games or games of low rank (i.e. having low rank on the sum of the two payoff matrices).

Another interesting question is whether one can have last-iterate convergence rates with game-independent. Both the OMWU method \citep{Wei2021LinearLC} and FLBR attain rates where the parameters depend on the payoff matrices and not just the dimension of the game.

\paragraph{Acknowledgments}
This work has been partially supported by project MIS 5154714 of the National Recovery and Resilience Plan Greece 2.0 funded by the European Union under the NextGenerationEU Program.
It has also been supported by the H.F.R.I call "Basic research Financing (Horizontal support of all Sciences)" under the National Recovery and Resilience Plan Greece 2.0 funded by the European Union - NextGenerationEU (H.F.R.I. Project Number: 15877) and it has been partially supported by the research
project ”Learning approaches for solution concepts in games
(LASCON)” of the internal grants 2022 of the ICS-FORTH. Also partially supported by the EU under Horizon Europe's TwinODIS project(No. 101160009).

\bibliographystyle{plainnat}
\bibliography{references}

@inproceedings{FMPV22,
  author       = {Michail Fasoulakis and
                  Evangelos Markakis and
                  Yannis Pantazis and
                  Constantinos Varsos},
  title        = {Forward Looking Best-Response Multiplicative Weights Update Methods
                  for Bilinear Zero-sum Games},
  booktitle    = {International Conference on Artificial Intelligence and Statistics,
                  {AISTATS} 2022},
  volume       = {151},
  pages        = {11096--11117},
  year         = {2022}
}

@inproceedings{BP18,
  author       = {James P. Bailey and
                  Georgios Piliouras},
  title        = {Multiplicative Weights Update in Zero-Sum Games},
  booktitle    = {Proceedings of the 2018 {ACM} Conference on Economics and Computation, 2018},
  pages        = {321--338},
  publisher    = {{ACM}},
  year         = {2018}
}

@inproceedings{Sokota+23,
  author       = {Samuel Sokota and
                  Ryan D'Orazio and
                  J. Zico Kolter and
                  Nicolas Loizou and
                  Marc Lanctot and
                  Ioannis Mitliagkas and
                  Noam Brown and
                  Christian Kroer},
  title        = {A Unified Approach to Reinforcement Learning, Quantal Response Equilibria,
                  and Two-Player Zero-Sum Games},
  booktitle    = {The Eleventh International Conference on Learning Representations,
                  {ICLR} 2023},
  year         = {2023}
}

@inproceedings{LiuOYZ23,
  author       = {Mingyang Liu and
                  Asuman E. Ozdaglar and
                  Tiancheng Yu and
                  Kaiqing Zhang},
  title        = {The Power of Regularization in Solving Extensive-Form Games},
  booktitle    = {The Eleventh International Conference on Learning Representations,
                  {ICLR} 2023},
  year         = {2023}
}

@inproceedings{FL0H25,
  author       = {Zijian Fang and
                  Zongkai Liu and
                  Chao Yu and
                  Chaohao Hu},
  title        = {Rapid Learning in Constrained Minimax Games with Negative Momentum},
  booktitle    = {AAAI-25, Sponsored by the Association for the Advancement of Artificial
                  Intelligence},
  pages        = {16541--16549},
  year         = {2025}
}

@inproceedings{MPP18,
  author       = {Panayotis Mertikopoulos and
                  Christos H. Papadimitriou and
                  Georgios Piliouras},
  title        = {Cycles in Adversarial Regularized Learning},
  booktitle    = {Proceedings of the Twenty-Ninth Annual {ACM-SIAM} Symposium on Discrete
                  Algorithms, {SODA} 2018},
  pages        = {2703--2717},
  publisher    = {{SIAM}},
  year         = {2018}
}

@inproceedings{GPMXWOCB14,
  author    = {Ian J. Goodfellow and Jean Pouget-Abadie and Mehdi Mirza and Bing Xu and David Warde-Farley and Sherjil Ozair and Aaron Courville and Yoshua Bengio},
  title     = {Generative {A}dversarial {N}ets},
  booktitle = {Proceedings of Annual Conference
               on Neural Information Processing Systems (NIPS '14)},
  pages     = {2672--2680},
  year      = {2014}
}

@inproceedings{COZ22,
  author    = {Cai, Yang and Oikonomou, Argyris and Zheng, Weiqiang},
  booktitle = {Advances in Neural Information Processing Systems},
  pages     = {33904--33919},
  title     = {Finite-Time Last-Iterate Convergence for Learning in Multi-Player Games},
  volume    = {35},
  year      = {2022}
  }

@inproceedings{GorbunovTG22,
  author    = {Gorbunov, Eduard and Taylor, Adrien and Gidel, Gauthier},
  booktitle = {Advances in Neural Information Processing Systems},
  pages     = {21858--21870},
  title     = {Last-Iterate Convergence of Optimistic Gradient Method for Monotone Variational Inequalities},
  volume    = {35},
  year      = {2022}
}

@inproceedings{Wei2021LinearLC,
  author    = {Chen-Yu Wei and Chung-Wei Lee and Mengxiao Zhang and Haipeng Luo},
  title     = {Linear Last-iterate Convergence in Constrained Saddle-point Optimization},
  booktitle = {Proceedings of the 9th International Conference on Learning Representations {ICLR} '21},
  year      = {2021}
}

@inproceedings{caiforget,
  author    = {Cai, Yang and Farina, Gabriele and Grand-Cl\'{e}ment, Julien and Kroer, Christian and Lee, Chung-Wei and Luo, Haipeng and Zheng, Weiqiang},
  booktitle = {Advances in Neural Information Processing Systems},
  pages     = {23406--23434},
  title     = {Fast Last-Iterate Convergence of Learning in Games Requires Forgetful Algorithms},
  volume    = {37},
  year      = {2024}
}

@article{CLS21,
  author  = {Michael B. Cohen and
             Yin Tat Lee and
             Zhao Song},
  title   = {Solving Linear Programs in the Current Matrix Multiplication Time},
  journal = {J. {ACM}},
  volume  = {68},
  number  = {1},
  pages   = {3:1--3:39},
  year    = {2021}
}

@inproceedings{BrandLLSS0W21,
  author    = {Jan van den Brand and
               Yin Tat Lee and
               Yang P. Liu and
               Thatchaphol Saranurak and
               Aaron Sidford and
               Zhao Song and
               Di Wang},
  title     = {Minimum cost flows, MDPs, and $L_1$ regression
               in nearly linear time for dense instances},
  booktitle = {53rd Annual {ACM} {SIGACT} Symposium on Theory of Computing (STOC '21)},
  pages     = {859--869},
  publisher = {{ACM}},
  year      = {2021}
}

@article{hoda2010smoothing,
  title     = {Smoothing techniques for computing {N}ash equilibria of sequential games},
  author    = {Hoda, Samid and Gilpin, Andrew and Pena, Javier and Sandholm, Tuomas},
  journal   = {Mathematics of Operations Research},
  volume    = {35},
  number    = {2},
  pages     = {494--512},
  year      = {2010},
  publisher = {INFORMS}
}

@article{gilpin2012first,
  title     = {First-order algorithm with convergence for-equilibrium in two-person zero-sum games},
  author    = {Gilpin, Andrew and Pena, Javier and Sandholm, Tuomas},
  journal   = {Mathematical programming},
  volume    = {133},
  number    = {1},
  pages     = {279--298},
  year      = {2012}
}

@article{R51,
  title   = {An iterative method of solving a game},
  author  = {Julia Robinson},
  journal = {Annals of Mathematics},
  volume  = {54},
  number  = {2},
  pages   = {296--301},
  year    = {1951}
}

@article{brown1951iterative,
  title   = {Iterative solution of games by fictitious play},
  author  = {Brown, George W},
  journal = {Act. Anal. Prod Allocation},
  volume  = {13},
  number  = {1},
  pages   = {374},
  year    = {1951}
}

@article{AHK12,
  author  = {Sanjeev Arora and
             Elad Hazan and
             Satyen Kale},
  title   = {The Multiplicative Weights Update Method: a Meta-Algorithm and Applications},
  journal = {Theory Comput.},
  volume  = {8},
  number  = {1},
  pages   = {121--164},
  year    = {2012}
}

@inproceedings{CJST19,
  author    = {Yair Carmon and
               Yujia Jin and
               Aaron Sidford and
               Kevin Tian},
  title     = {Variance Reduction for Matrix Games},
  booktitle = {Advances in Neural Information Processing Systems 32: Annual Conference
               on Neural Information Processing Systems},
  pages     = {11377--11388},
  year      = {2019}
}

@inproceedings{CJJS24,
  author    = {Yair Carmon and
               Arun Jambulapati and
               Yujia Jin and
               Aaron Sidford},
  title     = {A Whole New Ball Game: {A} Primal Accelerated Method for Matrix Games
               and Minimizing the Maximum of Smooth Functions},
  booktitle = {Proceedings of the 2024 {ACM-SIAM} Symposium on Discrete Algorithms,
               {SODA'24}},
  pages     = {3685--3723},
  publisher = {{SIAM}},
  year      = {2024}
}

@inproceedings{FarinaKS21,
  author    = {Gabriele Farina and
               Christian Kroer and
               Tuomas Sandholm},
  title     = {Faster Game Solving via Predictive Blackwell Approachability: Connecting
               Regret Matching and Mirror Descent},
  booktitle = {Thirty-Fifth {AAAI} Conference on Artificial Intelligence, {AAAI}},
  pages     = {5363--5371},
  year      = {2021}
}

@article{K76,
  title   = {The extragradient method for finding saddle points and other problems},
  author  = {Galina Korpelevich},
  journal = {Matecon},
  volume  = {12},
  pages   = {747--756},
  year    = {1976}
}

@article{P80,
  title   = {A modification of the {A}rrow-{H}urwicz method for search of saddle points},
  author  = {Leonid Denisovich Popov},
  journal = {Mathematical notes of the Academy of Sciences of the USSR},
  volume  = {28},
  pages   = {845--848},
  year    = {1980}
}

@inproceedings{daskalakis2018training,
  title     = {Training {GAN}s with Optimism},
  author    = {Constantinos Daskalakis and Andrew Ilyas and Vasilis Syrgkanis and Haoyang Zeng},
  booktitle = {Proceedings of the International Conference on Learning Representations (ICLR'18)},
  year      = {2018}
}

@inproceedings{DBLP:conf/aistats/LiangS19,
  author    = {Tengyuan Liang and
               James Stokes},
  title     = {Interaction Matters: {A} Note on Non-asymptotic Local Convergence
               of Generative Adversarial Networks},
  booktitle = {The 22nd International Conference on Artificial Intelligence and Statistics,
               {AISTATS} 2019},
  series    = {Proceedings of Machine Learning Research},
  volume    = {89},
  pages     = {907--915},
  year      = {2019}
}

@inproceedings{DBLP:conf/aistats/DiakonikolasDJ21,
  author    = {Jelena Diakonikolas and
               Constantinos Daskalakis and
               Michael I. Jordan},
  title     = {Efficient Methods for Structured Nonconvex-Nonconcave Min-Max Optimization},
  booktitle = {The 24th International Conference on Artificial Intelligence and Statistics,
               {AISTATS} 2021},
  volume    = {130},
  pages     = {2746--2754},
  year      = {2021}
}

@inproceedings{Azizian,
  author    = {Wa{\"{\i}}ss Azizian and
               Damien Scieur and
               Ioannis Mitliagkas and
               Simon Lacoste{-}Julien and
               Gauthier Gidel},
  title     = {Accelerating Smooth Games by Manipulating Spectral Shapes},
  booktitle = {The 23rd International Conference on Artificial Intelligence and Statistics,
               {AISTATS} 2020},
  volume    = {108},
  pages     = {1705--1715},
  year      = {2020}
}

@inproceedings{Daskalakis2019LastIterateCZ,
  author    = {Constantinos Daskalakis and
               Ioannis Panageas},
               title     = {Last-Iterate Convergence: Zero-Sum Games and Constrained Min-Max Optimization},
  booktitle = {10th Innovations in Theoretical Computer Science Conference, {ITCS}
               2019},
  volume    = {124},
  pages     = {27:1--27:18},
  year      = {2019}
}

@article{LY23,
  author     = {Haihao Lu and
                Jinwen Yang},
  title      = {On the Infimal Sub-differential Size of Primal-Dual Hybrid Gradient Method and Beyond},
  journal    = {CoRR},
  volume     = {abs/2206.12061},
  year       = {2023},
  eprinttype = {arXiv},
  pages      = {1--29}
}

@inproceedings{mokhtari2020unified,
  title     = {A unified analysis of extra-gradient and optimistic gradient methods for saddle point problems: Proximal point approach},
  author    = {Mokhtari, Aryan and Ozdaglar, Asuman and Pattathil, Sarath},
  booktitle = {International Conference on Artificial Intelligence and Statistics},
  pages     = {1497--1507},
  year      = {2020}
}

@inproceedings{HsiehIMM20,
  author    = {Hsieh, Yu-Guan and Iutzeler, Franck and Malick, J\'{e}r\^{o}me and Mertikopoulos, Panayotis},
  booktitle = {Advances in Neural Information Processing Systems},
  pages     = {16223--16234},
  title     = {Explore Aggressively, Update Conservatively: Stochastic Extragradient Methods with Variable Stepsize Scaling},
  volume    = {33},
  year      = {2020}
}

@inproceedings{GPD20,
  author    = {Golowich, Noah and Pattathil, Sarath and Daskalakis, Constantinos},
  booktitle = {Advances in Neural Information Processing Systems},
  pages     = {20766--20778},
  title     = {Tight last-iterate convergence rates for no-regret learning in multi-player games},
  volume    = {33},
  year      = {2020}
}

@inproceedings{DSZ21,
  author    = {Constantinos Daskalakis and
               Stratis Skoulakis and
               Manolis Zampetakis},
  title     = {The complexity of constrained min-max optimization},
  booktitle = {53rd Annual {ACM} {SIGACT} Symposium on Theory of Computing ({STOC} '21)},
  pages     = {1466--1478},
  publisher = {{ACM}},
  year      = {2021}
}

@inproceedings{PP24,
  author    = {Nikolas Patris and
               Ioannis Panageas},
  title     = {Learning Nash Equilibria in Rank-1 Games},
  booktitle = {Proceedings of the Twelfth International Conference on Learning Representations (ICLR'24)},
  year      = {2024}
}

@inproceedings{anagnostides2022last,
  author    = {Ioannis Anagnostides and
               Ioannis Panageas and
               Gabriele Farina and
               Tuomas Sandholm},
  title     = {On Last-Iterate Convergence Beyond Zero-Sum Games},
  booktitle = {International Conference on Machine Learning, {ICML} 2022},
  volume    = {162},
  pages     = {536--581},
  year      = {2022}
}

@book{van1991stability,
  title     = {Stability and perfection of Nash equilibria},
  author    = {Van Damme, Eric},
  volume    = {339},
  year      = {1991}
}

@article{NTHW21,
  author  = {Nakagawa, Kenji and Takei, Yoshinori and Hara, Shin-ichiro and Watabe, Kohei},
  title   = {Analysis of the Convergence Speed of the {A}rimoto-{B}lahut Algorithm by the Second-Order Recurrence Formula},
  year    = {2021},
  volume  = {67},
  number  = {10},
  journal = {IEEE Transactions on Information Theory},
  pages   = {6810--6831}
}

@inproceedings{PPSC23,
  author    = {Panageas, Ioannis and Patris, Nikolas and Skoulakis, Stratis and Cevher, Volkan},
  booktitle = {Advances in Neural Information Processing Systems},
  pages     = {32401--32423},
  title     = {Exponential Lower Bounds for Fictitious Play in Potential Games},
  volume    = {36},
  year      = {2023}
}

@article{freund1999adaptive,
  title     = {Adaptive game playing using multiplicative weights},
  author    = {Freund, Yoav and Schapire, Robert E},
  journal   = {Games and Economic Behavior},
  volume    = {29},
  number    = {1-2},
  pages     = {79--103},
  year      = {1999}
}

@article{nemirovski2004prox,
  title     = {Prox-method with rate of convergence {O}(1/t) for variational inequalities with Lipschitz continuous monotone operators and smooth convex-concave saddle point problems},
  author    = {Nemirovski, Arkadi},
  journal   = {SIAM Journal on Optimization},
  volume    = {15},
  number    = {1},
  pages     = {229--251},
  year      = {2004}
}

@article{ST04,
  author  = {Daniel A. Spielman and
             Shang{-}Hua Teng},
  title   = {Smoothed analysis of algorithms: Why the simplex algorithm usually
             takes polynomial time},
  journal = {J. {ACM}},
  volume  = {51},
  number  = {3},
  pages   = {385--463},
  year    = {2004}
}

@inproceedings{anagno_smooth,
  author    = {Anagnostides, Ioannis and Sandholm, Tuomas},
  booktitle = {Advances in Neural Information Processing Systems},
  pages     = {125839--125878},
  title     = {Convergence of ${\log}(1/\epsilon)$ for Gradient-Based Algorithms in Zero-Sum Games without the Condition Number: A Smoothed Analysis},
  volume    = {37},
  year      = {2024}
}

\appendix
\appendix
\crefalias{section}{appendix}
\section{Further related work}
\label{appsec:relwork}

We discuss here some additional research directions. 
Given the connection with linear programming, a variety  of algorithms focus on optimization and LP-based methods for zero-sum games. 
Theoretically, the best guarantees for solving the corresponding linear program can be found in \citet{CLS21} and \citet{BrandLLSS0W21}.
Regarding other methods, \citet{hoda2010smoothing} use Nesterov's first order smoothing techniques to achieve an $\varepsilon$-equilibrium in $O(1/\varepsilon)$ iterations, with the added benefits of simplicity and rather low computational cost per iteration. Following up on that work,  \citet{gilpin2012first} propose an iterated version of Nesterov's smoothing technique, which runs within $O(\frac{||A||}{\delta(A)}\cdot \ln(1/\varepsilon)) $ iterations. This is a significant improvement, with the caveat that the complexity depends on a condition measure $\delta(A)$, with $A$ being the payoff matrix. 

Furthermore, further results have been obtained regarding the design of algorithms with convergence guarantees for extensive form games. Although such games are not within the scope of our work, the techniques could prove useful for our restricted class of normal-form zero-sum games. Some of the main ideas that have been exploited in this literature concern regularization, see e.g. \cite{Sokota+23,LiuOYZ23} and negative momentum \cite{FL0H25}.

\section{Missing proofs from Section 2}
\begin{proof}[Proof for \Cref{lem:norm to Duality gap}]
We have that for any $i$, 
\begin{align*}
e^\top _iRy - e^\top _iRy^{*} &= \sum_j R_{ij}\cdot y_j - \sum_j R_{ij}\cdot y^{*}_j \\
&=  \sum_j R_{ij}\cdot (y_j-y^{*}_j) \\
&\leq \sum_j |R_{ij}\cdot (y_j-y^{*}_j)|\\
&=\sum_j R_{ij}\cdot |(y_j-y^{*}_j)| \\
&\leq \sum_j |(y_j-y^{*}_j)|\\
&=  ||y-y^{*}||_1.
\end{align*}
Thus, if $b = \arg \max_i e^\top _iRy$, then $\max_i e^\top _iRy = e^\top _bRy \leq ||y-y^{*}||_1 + e^\top _bRy^{*} \leq ||y-y^{*}||_1 + v$.
The second part of the lemma follows in a similar manner.
\end{proof}

\section{Forgetfullness}\label{app:forget}

Here, we extend their experiment of \cite{caiforget} and compare OGDA and FLBR. The hard game instance of \citet{caiforget} for OMWU, parameterized by $\delta \in (0,1)$, is the following:
\begin{equation}\label{game:ad}
A_\delta =  \begin{bmatrix}
    \frac{1}{2} + \delta & \frac{1}{2}\\[.1cm] 0 & 1
\end{bmatrix}.
\end{equation}
The game has a unique equilibrium $(x^*, y^*)$ where $x^*_1= \frac{1}{1+\delta}$ and $y_1^* = \frac{1}{2(1+\delta)}$. In \Cref{fig:cai_1}, we highlight the behavior of FLBR and OGDA, with $\delta=10^{-2}$. The upper subfigures show how the first coordinate of $x^t$ and $y^t$ vary over time, starting from the initialization $(x^0, y^0) = (1/2, 1/2)$. In the lower subfigures, we show the decrease in the duality gap over the iterations. Note that at the equilibrium, $x^*_1$ is close to 1, whereas $y^*_1$ is close to 1/2, and thus close to $y^0_1$. 
What we observe is that FLBR behaves similarly to OGDA in the sense that it forgets quickly, regarding the coordinate $x^t_1$, and therefore  it avoids slowdowns. Furthermore, FLBR does not overshoot $y^t_1$. It increases $y^t_1$ marginally before reaching the actual equilibrium point, whereas OGDA overshoots. This fact justifies the much faster convergence time of FLBR compared to OGDA, as seen in the lower subfigures.

\begin{figure}
\centering
\includegraphics[width=\linewidth]{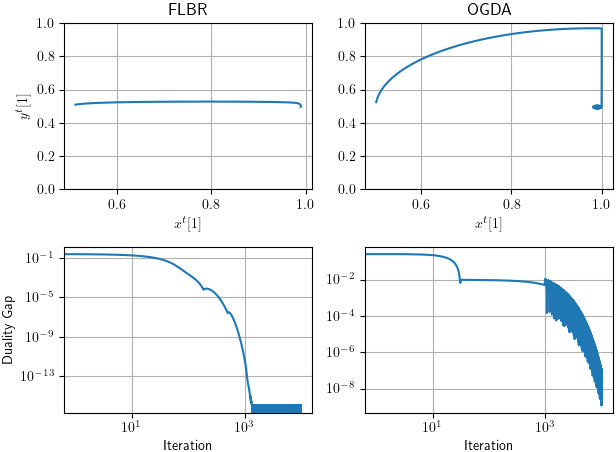}
\caption{FLBR vs OGDA in game $A_\delta$.}
\label{fig:cai_1}
\end{figure}
\end{document}